\documentclass[11pt]{article}

\newenvironment{marked}{}{}

\usepackage[top=1in, bottom=1in, left=1in, right=1in]{geometry}

\ifdefined\enforceshort\newcommand{\ONLYFULL}[1]{}\newcommand{\ONLYSHORT}[1]{#1}\else\newcommand{\ONLYFULL}[1]{#1}\newcommand{\ONLYSHORT}[1]{}\fi

\usepackage[american]{babel}
\usepackage{paralist}
\usepackage{wrapfig}

\usepackage{amsthm}
\usepackage{amsmath,amsfonts}
\usepackage{amssymb,bm,bbm}

\PassOptionsToPackage{hyphens}{url}
\usepackage[linktocpage,
            colorlinks=true,linkcolor=black,
            citecolor=black,urlcolor=black,
            pdfpagemode=UseNone,
            bookmarks=true,
            breaklinks,
            bookmarksopen=false,
            hyperfootnotes=false,
            pdfhighlight=/I,
            pdftitle={Frechet distance has no strongly subquadratic algorithms unless SETH fails},
            pdfsubject={FOCS},
            pdfauthor={Karl Bringmann}
            ]{hyperref}

\usepackage{thmtools, thm-restate}
\declaretheorem[numberwithin=section]{theorem}
\declaretheorem[sibling=theorem]{property}
\declaretheorem[sibling=theorem]{lemma}
\declaretheorem[sibling=theorem]{claim}

\usepackage{xspace}
\usepackage[latin1]{inputenc}
\usepackage{url}
\usepackage{tikz}
\usepackage{pgfplots,pgfplotstable,filecontents}
\usepackage{subfig}

\pgfplotsset{compat=newest}

\usepackage{float}
\floatstyle{ruled}
\newfloat{myalgorithm}{thp}{lop}
\floatname{myalgorithm}{Algorithm}

\usepackage{hyperref}
\usepackage{algorithm}
\usepackage[noend]{algpseudocode}

\usepackage[numbers,sort&compress,sectionbib]{natbib}
\usepackage[labelfont=bf,font=footnotesize,indention=0.3cm,margin=0cm]{caption}  

\usepackage{setspace,color}   
\newcounter{nummer}

\newcommand{\rot}{{\ensuremath{\textup{rot}}}}
\newcommand{\Tr}{{\ensuremath{\textup{tr}}}}

\newcommand{\sym}{{\ensuremath{\textup{sym}}}}

\newcommand{\nphi}{{\ensuremath{N}}}
\newcommand{\mphi}{{\ensuremath{M}}}
\newcommand{\ddfr}{\ensuremath{d_{\textup{dF}}}}
\newcommand{\dfr}{\ensuremath{d_{\textup{F}}}}

\newcommand{\Fr}{Fr\'echet }

\newcommand{\sat}{\ensuremath{\mathsf{sat}}}
\newcommand{\kSAT}{\ensuremath{\mathsf{k\text{-}SAT}}}
\newcommand{\threeSAT}{\ensuremath{3\mathsf{\text{-}SAT}}}
\newcommand{\CNFSAT}{\ensuremath{\mathsf{CNF\text{-}SAT}}}
\newcommand{\CNFSETH}{\ensuremath{\mathsf{SETH'}}}
\newcommand{\SETH}{\ensuremath{\mathsf{SETH}}}
\newcommand{\ETH}{\ensuremath{\mathsf{ETH}}}
\newcommand{\true}{\ensuremath{\mathsf{T}}}
\newcommand{\false}{\ensuremath{\mathsf{F}}}

\newcommand{\ORTHOG}{\ensuremath{\mathsf{Orthog}}}
\newcommand{\ORTHOGHYPO}{\ensuremath{\mathsf{OrthogHypothesis}}}

\newcommand{\R}{\mathbb{R}}

\newcommand{\N}{\mathbb{N}}
\newcommand{\Oh}{\mathcal{O}}

\newcommand{\IGNORE}[1]{}

\def\mod{\operatorname{mod}}
\def\min{\operatorname{min}}
\def\max{\operatorname{max}}

\newcommand{\propref}[1]{Property~\ref{prop:#1}}
\newcommand{\thmref}[1]{Theorem~\ref{thm:#1}}

\newcommand{\thmrefss}[3]{Theorems~\ref{thm:#1}, \ref{thm:#2}, and~\ref{thm:#3}}

\newcommand{\lemref}[1]{Lemma~\ref{lem:#1}}
\newcommand{\lemrefs}[2]{Lemmas~\ref{lem:#1} and~\ref{lem:#2}}

\newcommand{\secref}[1]{Section~\ref{sec:#1}}

\renewcommand{\epsilon}{\ensuremath{\varepsilon}}
\newcommand{\eps}{\ensuremath{\varepsilon}}

\newcommand{\temporary}[1]{}

\let\oldsqrt\sqrt
\def\hksqrt{\mathpalette\DHLhksqrt}
\def\DHLhksqrt#1#2{\setbox0=\hbox{$#1\oldsqrt{#2\,}$}\dimen0=\ht0
   \advance\dimen0-0.2\ht0
   \setbox2=\hbox{\vrule height\ht0 depth -\dimen0}%
   {\box0\lower0.4pt\box2}}
\renewcommand{\sqrt}{\hksqrt}

\renewcommand{\le}{\leqslant}
\renewcommand{\ge}{\geqslant}

\catcode`@=11
\def\nphantom{\v@true\h@true\nph@nt}
\def\nvphantom{\v@true\h@false\nph@nt}
\def\nhphantom{\v@false\h@true\nph@nt}
\def\nph@nt{\ifmmode\def\next{\mathpalette\nmathph@nt}%
  \else\let\next\nmakeph@nt\fi\next}
\def\nmakeph@nt#1{\setbox\z@\hbox{#1}\nfinph@nt}
\def\nmathph@nt#1#2{\setbox\z@\hbox{$\m@th#1{#2}$}\nfinph@nt}
\def\nfinph@nt{\setbox\tw@\null
  \ifv@ \ht\tw@\ht\z@ \dp\tw@\dp\z@\fi
  \ifh@ \wd\tw@-\wd\z@\fi \box\tw@}
\newcount\minute \newcount\hour \newcount\hourMins
\def\now{\minute=\time \hour=\time \divide \hour by 60 \hourMins=\hour \multiply\hourMins by 60
  \advance\minute by -\hourMins \zeroPadTwo{\the\hour}:\zeroPadTwo{\the\minute}}

\def\today{\the\year-\zeroPadTwo{\the\month}-\zeroPadTwo{\the\day}}
\def\zeroPadTwo#1{\ifnum #1<10 0\fi #1}

\newcommand{\myrho}{0.70710678}
\newcommand{\myd}{0.04}
\newcommand{\defpts}{%
  \coordinate (rTwo) at (-1/3+1/100,-1/2);
  \coordinate (cTwoTrueZero) at (0,-1/2+1/50);
  \coordinate (cTwoFalseZero) at (0,-1/2-1/50);
  \coordinate (cTwoTrueOne) at (1/3-1/100,-1/2+1/50);
  \coordinate (cTwoFalseOne) at (1/3-1/100,-1/2-1/50);
  \coordinate (sTwo) at (-1/3-1/50,0);
  \coordinate (tTwo) at (1/3+1/50,0);
  \coordinate (sTwoStar) at (-1/3-1/100,-4/5);
  \coordinate (tTwoStar) at (1/3+1/100,-4/5);
  \coordinate (rOne) at (-1/3,1/2);
  \coordinate (cOneTrueZero) at (0,1/2-1/50);
  \coordinate (cOneFalseZero) at (0,1/2+1/50);
  \coordinate (cOneTrueOne) at (1/3,1/2-1/50);
  \coordinate (cOneFalseOne) at (1/3,1/2+1/50);
  \coordinate (sOne) at (-1/3-1/50,1/5);
  \coordinate (tOne) at (1/3+1/50,1/5);
  \coordinate (UL1) at (-\myd,0);
  \coordinate (UL2) at (-\myrho+\myd,\myrho);
  \coordinate (UL3) at (-\myrho+\myd,3*\myrho);
  \coordinate (UL4) at (-\myrho+2*\myd,2*\myrho);
  \coordinate (UL5) at (-\myrho+2*\myd,\myrho+1*\myd);
  \coordinate (UR1) at (\myrho-2*\myd,\myrho+1*\myd);
  \coordinate (UR2) at (\myrho-2*\myd,2*\myrho);
  \coordinate (UR3) at (\myrho-\myd,3*\myrho);
  \coordinate (UR4) at (\myrho-\myd,\myrho);
  \coordinate (UR5) at (\myd,0);
}
\newcommand{\drawPOne}{
  \defpts;
  \node at (rOne) {};
  \node at (cOneTrueZero) {};
  \node at (cOneFalseZero) {};
  \node at (cOneTrueOne) {};
  \node at (cOneFalseOne) {};
  \node at (sOne) {};
  \node at (tOne) {};
  \path (sOne) edge (rOne)
        (rOne) edge (cOneFalseZero)
        (cOneFalseZero) edge (cOneTrueOne)
        (cOneTrueOne) edge (tOne)
        (tOne) edge (sOne)
        (sOne) edge (rOne)
        (rOne) edge (cOneTrueZero)
        (cOneTrueZero) edge (cOneFalseOne)
        (cOneFalseOne) edge (tOne);
}
\newcommand{\drawPTwo}{
  \defpts;
  \node at (rTwo) {};
  \node at (cTwoTrueZero) {};
  \node at (cTwoFalseZero) {};
  \node at (cTwoTrueOne) {};
  \node at (cTwoFalseOne) {};
  \node at (sTwo) {};
  \node at (tTwo) {};
  \node at (sTwoStar) {};
  \node at (tTwoStar) {};
  \path (sTwo) edge (sTwoStar)
        (sTwoStar) edge (rTwo)
        (rTwo) edge (cTwoTrueZero)
        (cTwoTrueZero) edge (cTwoFalseOne)
        (cTwoFalseOne) edge (rTwo)
        (rTwo) edge (cTwoFalseZero)
        (cTwoFalseZero) edge (cTwoTrueOne)
        (cTwoTrueOne) edge (tTwoStar)
        (tTwoStar) edge (tTwo);
}
\newcommand{\drawUL}{
  \defpts;
  \node at (UL1) {};
  \node at (UL2) {};
  \node at (UL3) {};
  \node at (UL4) {};
  \node at (UL5) {};
  \draw (UL1) -- (UL2) -- (UL3) -- (UL4) -- (UL5);
}
\newcommand{\drawUR}{
  \defpts;
  \node at (UR1) {};
  \node at (UR2) {};
  \node at (UR3) {};
  \node at (UR4) {};
  \node at (UR5) {};
  \draw (UR1) -- (UR2) -- (UR3) -- (UR4) -- (UR5);
}
\newcommand{\drawU}{
  \defpts;
  \drawUL;
  \drawUR;
  \draw (UL5) -- (UR1);
}

\title{
Why walking the dog takes time: \Fr distance has no strongly subquadratic algorithms unless SETH fails
}

\author{Karl Bringmann\thanks{Max Planck Institute for Informatics, Campus E1 4, 66123 Saarbr\"ucken, Germany; \texttt{karl.bringmann@mpi-inf.mpg.de}. Karl Bringmann is a recipient of the \emph{Google Europe Fellowship in Randomized Algorithms}, and this research is supported in part by this Google Fellowship.}
}
\begin{document}

\maketitle
%

\medskip

\begin{abstract}
The \Fr distance is a well-studied and very popular measure of similarity of two curves. Many variants and extensions have been studied since Alt and Godau introduced this measure to computational geometry in 1991. Their original algorithm to compute the \Fr distance of two polygonal curves with $n$ vertices has a runtime of $\Oh(n^2 \log n)$. More than 20 years later, the state of the art algorithms for most variants still take time more than $\Oh(n^2 / \log n)$, but no matching lower bounds are known, not even under reasonable complexity theoretic assumptions. 

To obtain a conditional lower bound, in this paper we assume the Strong Exponential Time Hypothesis or, more precisely, that there is no $\Oh^*((2-\delta)^N)$ algorithm for \CNFSAT\ for any $\delta > 0$. Under this assumption we show that the \Fr distance cannot be computed in strongly subquadratic time, i.e., in time $\Oh(n^{2-\delta})$ for any $\delta > 0$. This means that finding faster algorithms for the \Fr distance is as hard as finding faster \CNFSAT\ algorithms, and the existence of a strongly subquadratic algorithm can be considered unlikely.

Our result holds for both the continuous and the discrete \Fr distance. 
We extend the main result in various directions. Based on the same assumption we (1) show non-existence of a strongly subquadratic 1.001-approximation, (2) present tight lower bounds in case the numbers of vertices of the two curves are imbalanced, and (3) examine realistic input assumptions ($c$-packed curves).



\end{abstract}

\newpage

\section{Introduction} \label{sec:intro}

Intuitively, the (continuous) \Fr distance of two curves $P,Q$ is the minimal length of a leash required to connect a dog to its owner, as they walk along $P$ or $Q$, respectively, without backtracking. 
The \Fr distance is a very popular measure of similarity of two given curves. In contrast to distance notions such as the Hausdorff distance, it takes into account the order of the points along the curve, and thus better captures the similarity as perceived by human observers~\cite{Alt09}. 

Alt and Godau introduced the \Fr distance to computational geometry in 1991~\cite{AltG95,Godau91}. For polygonal curves $P$ and $Q$ with $n$ and $m$ vertices\footnote{We always assume that $m \le n$.}, respectively, they presented an $\Oh(n m \log (nm))$ algorithm. 
Since Alt and Godau's seminal paper, \Fr distance has become a rich field of research, with various directions
such as generalizations to surfaces (see, e.g.,~\cite{AltB10}), approximation algorithms for realistic input curves (\cite{AronovHPKW06,AltKW04,DriemelHPW12}), the geodesic and homotopic \Fr distance (see, e.g.,~\cite{ChambersETAL10,Wenk2010geodesic}), and many more variants (see, e.g.,~\cite{BuchinBW09,DriemelHP13,MaheshwariSSZ11,Indyk02}).
Being a natural measure for curve similarity, the \Fr distance has found applications in various areas such as signature verification (see, e.g.,~\cite{MunichP99}), map-matching tracking data (see, e.g.,~\cite{BrakatsoulasPSW05}), and moving objects analysis (see, e.g.,~\cite{BuchinBGLL11}).

A particular variant that we will also discuss in this paper is the \emph{discrete} \Fr distance. Here, intuitively the dog and its owner are replaced by two frogs, and in each time step each frog can jump to the next vertex along its curve or stay at its current vertex. Defined in~\cite{EiterM94}, the original algorithm for the discrete \Fr distance has runtime $\Oh(nm)$. 

Recently, improved algorithms have been found for some variants. Agarwal et al.~\cite{AgarwalBAKS13} showed how to compute the discrete \Fr distance in (mildly) subquadratic time $\Oh\big(nm \tfrac{\log \log n}{\log n}\big)$. Buchin et al.~\cite{BuchinBMM14} gave algorithms for the continuous \Fr distance that runs in time $\Oh(n^2 \sqrt{\log n} (\log \log n)^{3/2})$ on the Real RAM and $\Oh(n^2 (\log \log n)^2)$ on the Word RAM. However, the problem remains open whether there is a \emph{strongly subquadratic}\footnote{We use the term \emph{strongly subquadratic} to differentiate between this runtime and the \emph{(mildly) subquadratic} $\Oh(n^2 \log \log n / \log n)$ algorithm from~\cite{AgarwalBAKS13}.} algorithm for the \Fr distance, i.e., an algorithm with runtime $\Oh(n^{2-\delta})$ for any $\delta > 0$.
For a particular variant, the discrete \Fr distance with shortcuts, strongly subquadratic algorithms have been found recently~\cite{AvrahamFKKS13}, however, this seems to have no implications for the classical continuous or discrete \Fr distance.

The only known lower bound shows that the \Fr distance takes time $\Omega(n \log n)$ (in the algebraic decision tree model)~\cite{BuchinBKRW07}. 
The typical way of proving (conditional) quadratic lower bounds for geometric problems is via 3SUM~\cite{GajentaanO95}, in fact,
Alt conjectured that the \Fr distance is 3SUM-hard. Buchin et al.~\cite{BuchinBMM14} argued that the \Fr distance is unlikely to be 3SUM-hard, because it has strongly subquadratic decision trees. However, their argument breaks down in light of a recent result showing strongly subquadratic decision trees also for 3SUM~\cite{GronlundP14}. Hence, it is completely open whether the \Fr distance is 3SUM-hard.

\paragraph{Strong Exponential Time Hypothesis}
The Exponential Time Hypothesis (\ETH) and the Strong Exponential Time Hypothesis (\SETH), both introduced by Impagliazzo, Paturi, and Zane~\cite{ImpagliazzoPZ01,ImpagliazzoP01}, provide alternative ways of proving conditional lower bounds. ETH asserts that \threeSAT\ has no $2^{o(N)}$ algorithm, where $N$ is the number of variables, and can be used to prove matching lower bounds for a wealth of problems, see~\cite{LokshtanovMS11} for a survey. However, since this hypothesis does not specify the exact exponent, it is not suited for proving polynomial time lower bounds, where the exponent is important.

\begin{marked}
The stronger hypothesis \SETH\ asserts that there is no $\delta > 0$ such that \kSAT\ has an $\Oh((2-\delta)^N)$ algorithm for all~$k$. 
In this paper, we will use the following weaker variant, which has also been used in~\cite{PatrascuW10,RodittyVW13}.

\paragraph{Hypothesis \CNFSETH:} \emph{There is no $\Oh^*((2-\delta)^N)$ algorithm for \CNFSAT\ for any $\delta > 0$. Here, $\Oh^*$ hides polynomial factors in the number of variables $N$ and the number of clauses~$M$.}

\paragraph{}
While \SETH\ deals with formulas of width $k$, \CNFSETH\ deals with \CNFSAT, i.e., unbounded width clauses. Thus, it is a weaker assumption and more likely to be true. 
Note that exhaustive search takes time $\Oh^*(2^N)$, and the fastest known algorithms for \CNFSAT\ are only slighly faster than that, namely of the form $\Oh^*(2^{\nphi(1-C/\log(\mphi/\nphi))})$ for some positive constant $C$~\cite{DantsinH09,CalabroIP06}. Thus, \CNFSETH\ is a reasonable assumption that can be considered unlikely to fail.
It has been observed that one can use \SETH\ and \CNFSETH\ to prove lower bounds for polynomial time problems such as $k$-Dominating Set and others~\cite{PatrascuW10}, the diameter of sparse graphs~\cite{RodittyVW13}, and dynamic connectivity problems~\cite{AbboudVW14}. However, it seems to be applicable only for few problems, e.g., it seems to be a wide open problem to prove that 3SUM has no strongly subquadratic algorithms unless \SETH\ fails, similarly for matching, maximum flow, edit distance, and other classic problems.
\end{marked}

\paragraph{Main result}

Our main theorem gives strong evidence that the \Fr distance may have no strongly subquadratic algorithms by relating it to the Strong Exponential Time Hypothesis.

\begin{theorem} \label{thm:main}
  There is no $\Oh(n^{2-\delta})$ algorithm for the (continuous or discrete) \Fr distance for any $\delta > 0$, unless \CNFSETH\ fails.
\end{theorem}

\begin{marked}
Since \SETH\ and its weaker variant \CNFSETH\ are reasonable hypotheses, by this theorem one can consider it unlikely that the \Fr distance has strongly subquadratic algorithms. In particular, any strongly subquadratic algorithm for the \Fr distance would not only give improved algorithms for \CNFSAT\ that are much faster than exhaustive search, but also for various other problems such as Hitting Set, Set Splitting, and NAE-SAT via the reductions in~\cite{CyganETAL12}.
Alternatively, in the spirit of~\cite{PatrascuW10}, one can view the above theorem as a possible attack on \CNFSAT, as algorithms for the \Fr distance now could provide a route to faster \CNFSAT\ algorithms.
In any case, anyone trying to find strongly subquadratic algorithms for the \Fr distance should be aware that this is as hard as finding improved \CNFSAT\ algorithms, which might be impossible.
\end{marked}

We remark that all our lower bounds (unless stated otherwise) hold in the Euclidean plane, and thus also in $\R^d$ for any $d \ge 2$.

\paragraph{Extensions}
We extend our main result in two important directions: We show approximation hardness and we prove that the lower bound still holds for restricted classes of curves.

First, it would be desirable to have good approximation algorithms in strongly subquadratic time, say a near-linear time approximation scheme. We exclude such algorithms by proving that there is no 1.001-approximation for the \Fr distance in strongly subquadratic time unless \CNFSETH\ fails. 
Hence, within $n^{o(1)}$-factors any 1.001-approximation takes as much time as an exact algorithm. We did not try to optimize the constant 1.001, but only to find the asymptotically largest possible approximation ratio, which seems to be a constant.
We leave it as an open problem whether there is a strongly subquadratic $\Oh(1)$-approximation. The literature so far contains no strongly subquadratic approximation algorithms for general curves at all.

Second, it might be conceivable that if one curve has much fewer vertices than the other, i.e., $m \ll n$, then after some polynomial preprocessing on the smaller curve we can compute the \Fr distance of the two curves quickly, e.g., in total time $\Oh((n + m^3) \log n)$. 
Note that such a runtime is not ruled out by the trivial argument that any algorithm needs time $\Omega(n+m)$ for reading the input, and is also not ruled out by \thmref{main}, since the runtime is not subquadratic for $n=m$. We rule out such runtimes by proving that there is no $\Oh((nm)^{1-\delta})$ algorithm ``for any $m$'', unless \CNFSETH\ fails. More precisely, we prove this lower bound for the ``special case'' $m \approx n^\gamma$ for any constant $0 \le \gamma \le 1$. To make this formal, for any input parameter~$\alpha$ and constants $\gamma_0 < \gamma_1$ in $\R \cup \{-\infty,\infty\}$, we
say that a statement holds \emph{for any polynomial restriction of $n^{\gamma_0} \le \alpha \le n^{\gamma_1}$} if it holds restricted to instances with $n^{\gamma-\delta} \le \alpha \le n^{\gamma+\delta}$ for any constants $\delta > 0$ and $\gamma_0 + \delta \le \gamma \le \gamma_1 - \delta$. We obtain the following extension of the main result \thmref{main}, which yields tight lower bounds for any behaviour of $m$ and any $(1+\eps)$-approximation with $0 \le \eps \le 0.001$.

\begin{theorem} \label{thm:const}
  There is no $1.001$-approximation with runtime $\Oh((nm)^{1-\delta})$ for the (continuous or discrete) \Fr distance for any $\delta > 0$, unless \CNFSETH\ fails. 
  This holds for any polynomial restriction of $1 \le m \le n$.
\end{theorem}

%

\paragraph{Realistic input curves}
In attempts to capture the properties of realistic input curves,
strongly subquadratic algorithms have been devised for restricted classes of inputs such as backbone curves~\cite{AronovHPKW06}, $\kappa$-bounded and $\kappa$-straight~\cite{AltKW04}, and $\phi$-low density curves~\cite{DriemelHPW12}.
The most popular model are \emph{$c$-packed curves}, which have been used for various generalizations of the \Fr distance~\cite{ChenDGNW11,HarPeledR11,DriemelHP13}.
Driemel et al.~\cite{DriemelHPW12} introduced this model and presented a $(1+\eps)$-approximation for the continuous \Fr distance that runs in time $\Oh(cn/\eps + cn \log n)$, which works in any $\R^d$, $d \ge 2$.

While the algorithm of~\cite{DriemelHPW12} is near-linear for small $c$ and $1/\eps$, is is not clear whether its dependence on $c$ and $1/\eps$ is optimal for $c$ and $1/\eps$ that grow with $n$.
We give strong evidence that the algorithm of~\cite{DriemelHPW12} has optimal dependence on $c$ for any constant $0 < \eps \le 0.001$.

\begin{theorem} \label{thm:cpacked}
  There is no $1.001$-approximation with runtime $\Oh((cn)^{1-\delta})$ for the (continuous or discrete) \Fr distance on $c$-packed curves for any $\delta>0$, unless \CNFSETH\ fails. 
  This holds for any polynomial restriction of $1 \le c \le n$.
\end{theorem}

Since we prove this claim for any polynomial restriction $c \approx n^\gamma$, the above result excludes $1.001$-approximations with runtime, say, $\Oh(c^2 + n)$.

Regarding the dependence on $\eps$, in any dimension $d \ge 5$ we can prove a conditional lower bound that matches the dependency on $\eps$ of~\cite{DriemelHPW12} up to a polynomial.

\begin{theorem} \label{thm:cpackedFive}
  There is no $(1+\eps)$-approximation for the (continuous or discrete) \Fr distance on $c$-packed curves in $\R^d$, $d \ge 5$, with runtime $\Oh(\min\{cn/\sqrt{\eps},n^2\}^{1-\delta})$ for any $\delta>0$, unless \CNFSETH\ fails. 
  This holds for sufficiently small $\eps > 0$ and any polynomial restriction of $1 \le c \le n$ and $\eps \le 1$.
\end{theorem}

\begin{marked}
\paragraph{Outline of the main result}
To prove the main result we present a reduction from \CNFSAT\ to the \Fr distance. Given a \CNFSAT\ instance $\varphi$, we partition its variables into sets $V_1,V_2$ of equal size. In order to find a satisfying assignment of $\varphi$ we have to choose (partial) assignments $a_1$ of $V_1$ and $a_2$ of $V_2$. 
We will construct curves $P_1,P_2$ where $P_k$ is responsible for choosing $a_k$. To this end, $P_k$ consists of \emph{assignment gadgets}, one for each assignment of~$V_k$. Assignment gadgets are built of \emph{clause gadgets}, one for each clause. 
The assignment gadgets of assignments $a_1$ of $V_1$ and $a_2$ of $V_2$ are constructed such that they have \Fr distance at most 1 if and only if $(a_1,a_2)$ forms a satisfying assignment of $\varphi$. 
In $P_1$ and $P_2$ we connect these assignment gadgets with some additional curves to implement an OR-gadget, which forces any traversal of $(P_1,P_2)$ to walk along two assignment gadgets in parallel. If $\varphi$ is not satisfiable, then any pair of assignment gadgets has \Fr distance larger than 1, so that $P_1,P_2$ have \Fr distance larger than 1. If, on the other hand, a satisfying assignment $(a_1,a_2)$ of $\varphi$ exists, then we ensure that there is a traversal of $P_1,P_2$ that essentially only traverses the assignment gadgets of $a_1$ and $a_2$ in parallel, so that it always stays in distance 1. 

To argue about the runtime, since $P_k$ contains an assignment gadget for every assignment of one half of the variables, and every assignment gadget has polynomial size in $\mphi$, there are $n = \Oh^*(2^{\nphi/2})$ vertices on each curve. Thus, any $\Oh(n^{2-\delta})$ algorithm for the \Fr distance would yield an $\Oh^*(2^{(1-\delta/2) \nphi})$ algorithm for \CNFSAT, contradicting \CNFSETH.
\end{marked}

\paragraph{Remark: Orthogonal Vectors}
Let \ORTHOG\ be the problem of `finding a pair of orthogonal vectors'': given two sets $S_1,S_2 \subseteq \{0,1\}^d$ of $n$ vectors each, determine if there are $u \in S_1$ and $v \in S_2$ with $\langle u,v \rangle = \sum_{i=1}^d u_i v_i = 0$, where the sum is computed over the integers, see~\cite{Williams04,WilliamsY14}. Clearly, \ORTHOG\ can be solved in time $\Oh(n^2 d)$. However, \ORTHOG\ has no strongly subquadratic algorithms unless \CNFSETH\ fails. More precisely, in~\cite{Williams04} it was shown that \CNFSETH\ implies the following statement.

\paragraph{\ORTHOGHYPO:} \emph{There is no algorithm for \ORTHOG\ with runtime $\Oh(n^{2-\delta} d^{\Oh(1)})$ for any $\delta > 0$.}

\paragraph{}
All known conditional lower bounds based on \CNFSETH\ implicitly go through \ORTHOG\ or some variant of this problem. In fact, this is also the case for our results, as is easily seen by going through the proof in~\cite{Williams04} and noting that we use the same tricks. Specifically, given a \CNFSAT\ instance $\phi$ on variables $x_1,\ldots,x_N$ and clauses $C_1,\ldots,C_M$ we split the variables into two halves $V_1,V_2$ of equal size and enumerate all assignments $A_k$ of true and false to $V_k$. Then every clause $C_i$ specifies sets $B_k^i \subseteq A_k$ of partial assignments that do not make $C_i$ become true. Clearly, a satisfying assignment $(a_1,a_2) \in A_1 \times A_2$ has to evade $B_1^i \times B_2^i$ for all $i$. This problem is equivalent to an instance of \ORTHOG\ with $d=M$ and $n=2^{N/2}$, where $S_k$ contains a vector for every partial assignment $a_k \in A_k$ and the $i$-th position of this vector is 1 or 0, depending on whether $a_k \in B_k^i$ or not. In our proof, we could replace this instance by an arbitrary instance of \ORTHOG, yielding a reduction from \ORTHOG\ to the \Fr distance.

Hence, in \thmrefss{main}{cpacked}{cpackedFive} we could replace the assumption ``unless \CNFSETH\ fails'' by the weaker assumption ``unless \ORTHOGHYPO\ fails''. This is a stronger statement, since there is only more reason to believe that \ORTHOG\ has no strongly subquadratic algorithms than that there is for believing that \CNFSAT\ takes time $2^{N-o(N)}$. Moreover, it shows a relation between two polynomial time problems, \ORTHOG\ and the \Fr distance.

For \thmref{const} we would need an imbalanced version of the \ORTHOGHYPO, where the two sets $S_1,S_2$ have different sizes $n_1,n_2$. Then unless \CNFSETH\ fails there is no $\Oh((n_1 n_2)^{1-\delta} d^{\Oh(1)})$ algorithm for any $\delta > 0$, and this holds for any polynomial restriction of $1 \le n_1 \le n_2$, which follows from a slight generalization of~\cite{Williams04}. If we state this implication of \CNFSETH\ as a hypothesis \ORTHOGHYPO$^*$, then in \thmref{const} we could replace ``unless \CNFSETH\ fails'' by the weaker assumption ``unless \ORTHOGHYPO$^*$ fails''.

\paragraph{Organization}
We start by defining the variants of the \Fr distance, $c$-packedness, and other basic notions in \secref{preliminaries}.
\secref{general} deals with general curves. We prove the main result for the discrete \Fr distance in less than 3 pages in \secref{discrete}. This construction also already proves inapproximability. We generalize the proof to the continuous \Fr distance in \secref{nondiscrete} (which is more tedious than in the discrete case) and to $m \ll n$ in \secref{nm} (which is an easy trick).
\secref{cpacked} deals with $c$-packed curves. In \secref{cpackedconst} we present a new OR-gadget that generates less packed curves; plugging in the curves constructed in the main result proves \thmref{cpacked}. In \secref{cpackedFive} we make use of the fact that in $\ge4$ dimensions there are point sets $Q_1,Q_2$ of arbitrary size with each pair of points $(q_1,q_2)$ having distance exactly~1. This allows to construct less packed curves that we plug into the OR-gadget from the preceding section to prove \thmref{cpackedFive}.

\section{Preliminaries} \label{sec:preliminaries}

For $N \in \N$ we let $[N] := \{1,\ldots,N\}$.
A (polygonal) curve $P$ is defined by its vertices $p_1,\ldots,p_n$. We view $P$ as a continuous function $P \colon [0,n] \to \R^d$ with $P(i+\lambda) = (1-\lambda)p_i + \lambda p_{i+1}$ for $i \in [n-1]$, $\lambda \in [0,1]$. We write $|P| = n$ for the number of vertices of $P$.
For two curves $P_1,P_2$ we let $P_1 \circ P_2$ be the curve on $|P_1|+|P_2|$ vertices that first follows $P_1$, then walks along the segment from $P_1(|P_1|)$ to $P_2(0)$, and then follows $P_2$. In particular, for two points $p,q \in \R^d$ the curve $p \circ q$ is the segment from $p$ to $q$, and any curve $P$ on vertices $p_1,\ldots,p_n$ can be written as $P = p_1\circ \ldots \circ p_n$. 

Consider a curve $P$ and two points $p_1 = P(\lambda_1)$, $p_2 = P(\lambda_2)$ with $\lambda_1, \lambda_2 \in [0,n]$. We say that \emph{$p_1$ is within distance $D$ of $p_2$ along $P$} if the length of the subcurve of $P$ between $P(\lambda_1)$ and $P(\lambda_2)$ is at most $D$.


\paragraph{Variants of the \Fr distance}
Let $\Phi_n$ be the set of all continuous and non-decreasing functions $\phi$ from $[0,1]$ onto $[0,n]$. The \emph{continuous \Fr distance} between two curves $P_1,P_2$ with $|P_1|=n$, $|P_2|=m$ is defined as
$$ \dfr(P_1,P_2) := \inf_{\substack{\phi_1 \in \Phi_n \\\phi_2 \in \Phi_m}} \max_{t \in [0,1]} \|P_1(\phi_1(t)) - P_2(\phi_2(t))\|, $$
where $\|.\|$ denotes the Euclidean distance. We call $(\phi_1,\phi_2)$ a (continuous) \emph{traversal} of $(P_1,P_2)$, and say that it has \emph{width} $D$ if $\max_{t \in [0,1]} \|P_1(\phi_1(t)) - P_2(\phi_2(t))\| \le D$.

In the discrete case, we let $\Delta_n$ be the set of all non-decreasing functions $\phi$ from $[0,1]$ onto~$[n]$. The \emph{discrete \Fr distance} between two curves $P_1,P_2$ with $|P_1|=n$, $|P_2|=m$ is then defined as
$$ \ddfr(P_1,P_2) := \inf_{\substack{\phi_1 \in \Delta_n \\\phi_2 \in \Delta_m}} \max_{t \in [0,1]} \|P_1(\phi_1(t)) - P_2(\phi_2(t))\|. $$
We obtain an analogous notion of a (discrete) \emph{traversal} and its \emph{width}. Note that any $\phi \in \Delta_n$ is a staircase function attaining all values in $[n]$. Hence, $(\phi_1(t),\phi_2(t))$ changes only at finitely many points in time $t$. At any such \emph{time step} we jump to the next vertex in $P_1$ or $P_2$ or both.

It is known that for any curves $P_1,P_2$ we have $\dfr(P_1,P_2) \le \ddfr(P_1,P_2)$~\cite{EiterM94}.

\paragraph{Realistic input curves}
As an example of input restrictions that resemble practical input curves we consider the model of~\cite{DriemelHPW12}. A curve $P$ is \emph{$c$-packed} if for any point $q \in \R^d$ and any radius $r>0$ the total length of $P$ inside the ball $B(q,r)$ is at most $c r$. Here, $B(q,r)$ is the ball of radius $r$ around $q$. In this paper, we say that a curve $P$ is \emph{$\Theta(c)$-packed}, if there are constants $\alpha > \beta > 0$ such that $P$ is $\alpha c$-packed but not $\beta c$-packed.

This model is well motivated from a practical point of view. Examples of classes of $c$-packed curves are boundaries of convex polygons and $\gamma$-fat shapes as well as algebraic curves of bounded maximal degree (see~\cite{DriemelHPW12}).

\begin{marked}
\paragraph{Satisfiability} 
In \CNFSAT\ we are given a formula $\varphi$ on variables $x_1,\ldots,x_\nphi$ and clauses $C_1,\ldots,C_\mphi$ in conjunctive normal form with unbounded clause width.  
Let $V$ be any subset of the variables of $\varphi$. Let $a$ be any assignment of $\true$ (true) or $\false$ (false) to the variables of $V$. We call $a$ a \emph{partial assignment} and say that $a$ \emph{satisfies} a clause $C = \bigvee_{i \in I} x_i \vee \bigvee_{i \in J} \neg x_i$ if for some $i \in I \cap V$ we have $a(x_i) = \true$ or for some $i \in J \cap V$ we have $a(x_i) = \false$. We denote by $\sat(a,C)$ whether partial assignment $a$ satisfies clause $C$. Note that assignments $a$ of $V$ and $a'$ of the remaining variables $V'$ form a satisfying assignment $(a,a')$ of $\varphi$ if and only if we have $\sat(a,C_i) \vee \sat(a',C_i) = \true$ for all $i \in \{1,\ldots,\mphi\}$.

All bounds that we prove in this paper assume the hypothesis \CNFSETH\ (see \secref{intro}), which asserts that \CNFSAT\ has no $\Oh^*((2-\delta)^N)$ algorithm for any $\delta > 0$. Here, $\Oh^*$ hides polynomials factors in $\nphi$ and $\mphi$. The following is an easy corollary of \CNFSETH.

\begin{lemma} \label{lem:seth}
  There is no $\Oh^*((2-\delta)^\nphi)$ algorithm for \CNFSAT\ restricted to formulas with $\nphi$ variables and $\mphi \le 2^{\delta' \nphi}$ clauses for any $\delta, \delta' > 0$, unless \CNFSETH\ fails.
\end{lemma}
\begin{proof}
  Any such algorithm would imply an $\Oh^*((2-\delta)^\nphi)$ algorithm for \CNFSAT\ (with no restrictions on the input), since for $M \le 2^{\delta' \nphi}$ we can run the given algorithm, while for $M > 2^{\delta' \nphi}$ we can decide satisfiability in time $\Oh(M 2^\nphi) = \Oh(M^{1+1/\delta'}) = \Oh^*(1)$.
\end{proof}
\end{marked}

\section{General curves} \label{sec:general}

We first present a reduction from \CNFSAT\ to the \Fr distance and show that it proves \thmref{main} for the discrete \Fr distance. In \secref{nondiscrete} we then show that the same construction also works for the continuous \Fr distance. Finally, in \secref{nm} we generalize these results to curves with imbalanced numbers of vertices $n,m$ to show \thmref{const}.

\subsection{The basic reduction, discrete case} \label{sec:discrete}

Let $\varphi$ be a given \CNFSAT\ instance with variables $x_1,\ldots,x_{\nphi}$ and clauses $C_1,\ldots,C_{\mphi}$. We split the variables into two halves $V_1 := \{x_1,\ldots,x_{\nphi/2}\}$ and $V_2 := \{x_{\nphi/2+1},\ldots,x_{\nphi}\}$. For $k \in \{1,2\}$ let $A_k$ be all assignments\footnote{In later sections we will replace $V_1,V_2$ by different partitionings and $A_1,A_2$ by subsets of all assignments. The lemmas in this section are proven in a generality that allows this extension.} of \true\ or \false\ to the variables in $V_k$, so that $|A_k| = 2^{\nphi/2}$. In the whole section we let $\eps := 1/1000$.

We will construct two curves $P_1,P_2$ such that $\ddfr(P_1,P_2) \le 1$ if and only if $\varphi$ is satisfiable. In the construction we will use gadgets as follows.

\begin{wrapfigure}{r}{0.2\textwidth}
\vspace{-10pt}
\begin{center}
\begin{tikzpicture}[every node/.style={fill, circle, inner sep = 1pt},scale=3]
  \node (rTwo) at (-1/3,-1/2) {};
  \node (cTwoTrueZero) at (0,-1/2+1/50) {};
  \node (cTwoFalseZero) at (0,-1/2-1/50) {};
  \node (cTwoTrueOne) at (1/3,-1/2+1/50) {};
  \node (cTwoFalseOne) at (1/3,-1/2-1/50) {};
  
  \node (rOne) at (-1/3,1/2) {};
  \node (cOneTrueZero) at (0,1/2-1/50) {};
  \node (cOneFalseZero) at (0,1/2+1/50) {};
  \node (cOneTrueOne) at (1/3,1/2-1/50) {};
  \node (cOneFalseOne) at (1/3,1/2+1/50) {};
  
  \node[shape=coordinate,label=below:$c_{2,\false}^1$] at (cTwoFalseOne) {};
  \node[shape=coordinate,label=above:$c_{2,\true}^1$] at (cTwoTrueOne) {};
  \node[shape=coordinate,label=below:$c_{2,\false}^0$] at (cTwoFalseZero) {};
  \node[shape=coordinate,label=above:$c_{2,\true}^0$] at (cTwoTrueZero) {};
  \node[shape=coordinate,label=above:$r_2$] at (rTwo) {};
  
  \node[shape=coordinate,label=above:$c_{1,\false}^1$] at (cOneFalseOne) {};
  \node[shape=coordinate,label=below:$c_{1,\true}^1$] at (cOneTrueOne) {};
  \node[shape=coordinate,label=above:$c_{1,\false}^0$] at (cOneFalseZero) {};
  \node[shape=coordinate,label=below:$c_{1,\true}^0$] at (cOneTrueZero) {};
  \node[shape=coordinate,label=below:$r_1$] at (rOne) {};
\end{tikzpicture}
\end{center}
\vspace{-10pt}
\end{wrapfigure}

\paragraph{Clause gadgets} This gadget encodes whether a partial assignment satisfies a clause. We set for $i \in \{0,1\}$
\begin{align*}
  &c_{1,\true}^i := \big(i/3, \tfrac12 - \eps\big), \quad \;\;\, c_{1,\false}^i := \big(i/3, \tfrac12 + \eps\big),  \\
  &c_{2,\true}^i := \big(i/3, -\tfrac12 + \eps\big), \quad c_{2,\false}^i := \big(i/3, -\tfrac12 - \eps\big).
\end{align*}
Let $k \in \{1,2\}$. For any partial assignment $a_k \in A_k$ and clause $C_i$, $i \in [\mphi]$, we construct a clause gadget consisting of a single point,
$$ CG(a_k,i) := c_{k,\sat(a_k,C_i)}^{i \mod 2}. $$
Thus, if assignment $a_k$ satisfies clause $C_i$ then the corresponding clause gadget is nearer to the clause gadgets associated with $A_{3-k}$.
Explicitly calculating all pairwise distances of these points, we obtain the following lemma.

\begin{lemma} \label{lem:CG}
  Let $a_k \in A_k$, $k \in \{1,2\}$, and $i,j \in [\mphi]$. If $i \equiv j \pmod 2$ and $\sat(a_1,C_i) \vee \sat(a_2,C_j) = \true$ then $\|CG(a_1,i) - CG(a_2,j)\| \le 1$. Otherwise $\|CG(a_1,i) - CG(a_2,j)\| \ge 1+2\eps$.
\end{lemma}

\begin{wrapfigure}{r}{0.2\textwidth}
\vspace{-10pt}
\begin{center}
\begin{tikzpicture}[every node/.style={fill, circle, inner sep = 1pt},scale=3]
  \node (rTwo) at (-1/3,-1/2) {};
  \node (cTwoTrueZero) at (0,-1/2+1/50) {};
  \node (cTwoFalseZero) at (0,-1/2-1/50) {};
  \node (cTwoTrueOne) at (1/3,-1/2+1/50) {};
  \node (cTwoFalseOne) at (1/3,-1/2-1/50) {};
  
  \node (rOne) at (-1/3,1/2) {};
  \node[fill=gray!45] (cOneTrueZero) at (0,1/2-1/50) {};
  \node (cOneFalseZero) at (0,1/2+1/50) {};
  \node (cOneTrueOne) at (1/3,1/2-1/50) {};
  \node[fill=gray!45] (cOneFalseOne) at (1/3,1/2+1/50) {};
  
  \path[->] (rTwo) edge (cTwoTrueZero)
        (cTwoTrueZero) edge (cTwoFalseOne)
        (cTwoFalseOne) edge (cTwoFalseZero)
        (cTwoFalseZero) edge (cTwoTrueOne);
  
  \path[->,dotted] (rOne) edge (cOneFalseZero)
        (cOneFalseZero) edge (cOneTrueOne);
  
\end{tikzpicture}
\end{center}
\vspace{-10pt}
\end{wrapfigure}

\paragraph{Assignment gadgets} This gadget consists of clause gadgets and encodes the set of satisfied clauses for an assignment. We set 
$$ r_1 := (-\tfrac13, \tfrac12), \quad r_2 := (-\tfrac13, -\tfrac12). $$
The assignment gadget for any $a_k \in A_k$ consists the starting point $r_k$ followed by all clause gadgets of $a_k$,
$$ AG(a_k) := r_k \circ \bigcirc_{i \in [\mphi]} CG(a_k,i), $$
(recall the definition of $\circ$ in \secref{preliminaries}).
The figure to the right shows an assignment gadget on $\mphi=2$ clauses at the top and an assignment gadget on $\mphi=4$ clauses at the bottom. The arrows indicate the order in which the segments are traversed.

\begin{lemma} \label{lem:AG}
  Let $a_k \in A_k$, $k \in \{1,2\}$. If $(a_1,a_2)$ is a satisfying assignment of $\varphi$ then $\ddfr(AG(a_1),AG(a_2)) \le 1$. 
  If $(a_1,a_2)$ is not satisfying then $\ddfr(AG(a_1),AG(a_2)) > 1+\eps$, and we even have $\ddfr(AG(a_1) \circ \pi_1, AG(a_2) \circ \pi_2) > 1+\eps$ for any curves $\pi_1,\pi_2$.
\end{lemma}
\begin{proof}
  If $(a_1,a_2)$ is satisfying then the parallel traversal 
  $$(r_1,r_2),(CG(a_1,1),CG(a_2,1)),\ldots,(CG(a_1,\mphi),CG(a_2,\mphi))$$
  has width 1 by \lemref{CG}. 
  
  Assume for the sake of contradiction that $(a_1,a_2)$ is not satisfying but there is a traversal of $(AG(a_1) \circ \pi_1, AG(a_2) \circ \pi_2)$ with width $1+\eps$.
  Observe that $\|r_1 - r_2\| = 1$ and $\|r_k - c_{3-k,x}^i\| \ge 1+2\eps$ for any $k \in \{1,2\}, i \in \{0,1\}, x \in \{\true,\false\}$.
  Thus, the traversal has to start at positions $(r_1,r_2)$ and then step to positions $(CG(a_1,1),CG(a_2,1))$, as advancing in only one of the curves leaves us in distance larger than $1+\eps$. Inductively and using \lemref{CG}, the same argument shows that in the $i$-th step we are at positions $(CG(a_1,i),CG(a_2,i))$ for any $i \in [\mphi]$. Since there is an unsatisfied clause $C_i$, so that $\|CG(a_1,i) - CG(a_2,i)\| \ge 1+2\eps$ by \lemref{CG}, we obtain a contradiction.
\end{proof}

\begin{wrapfigure}{r}{0.23\textwidth}
\vspace{-0pt}
\begin{center}
\begin{tikzpicture}[every node/.style={fill, circle, inner sep = 1pt},scale=3]
  \node (rTwo) at (-1/3+1/100,-1/2) {};
  \node (cTwoTrueZero) at (0,-1/2+1/50) {};
  \node (cTwoFalseZero) at (0,-1/2-1/50) {};
  \node (cTwoTrueOne) at (1/3-1/100,-1/2+1/50) {};
  \node (cTwoFalseOne) at (1/3-1/100,-1/2-1/50) {};
  \node (sTwo) at (-1/3-1/50,0) {};
  \node (tTwo) at (1/3+1/50,0) {};
  \node (sTwoStar) at (-1/3-1/100,-4/5) {};
  \node (tTwoStar) at (1/3+1/100,-4/5) {};
  
  \node (rOne) at (-1/3,1/2) {};
  \node (cOneTrueZero) at (0,1/2-1/50) {};
  \node (cOneFalseZero) at (0,1/2+1/50) {};
  \node (cOneTrueOne) at (1/3,1/2-1/50) {};
  \node (cOneFalseOne) at (1/3,1/2+1/50) {};
  \node (sOne) at (-1/3-1/50,1/5) {};
  \node (tOne) at (1/3+1/50,1/5) {};
  
  \path[->] (sTwo) edge (sTwoStar)
        (sTwoStar) edge (rTwo)
        (rTwo) edge (cTwoTrueZero)
        (cTwoTrueZero) edge (cTwoFalseOne)
        (cTwoFalseOne) edge (rTwo)
        (rTwo) edge (cTwoFalseZero)
        (cTwoFalseZero) edge (cTwoTrueOne)
        (cTwoTrueOne) edge (tTwoStar)
        (tTwoStar) edge (tTwo);
  
  \path[->,dotted] (sOne) edge (rOne)
        (rOne) edge (cOneFalseZero)
        (cOneFalseZero) edge (cOneTrueOne)
        (cOneTrueOne) edge (tOne)
        (tOne) edge (sOne)
        (sOne) edge (rOne)
        (rOne) edge (cOneTrueZero)
        (cOneTrueZero) edge (cOneFalseOne)
        (cOneFalseOne) edge (tOne);
  
  \node[shape=coordinate,label=left:$s_1$] at (sOne) {};
  \node[shape=coordinate,label=left:$s_2$] at (sTwo) {};
  \node[shape=coordinate,label=left:$s_2^*$] at (sTwoStar) {};
  \node[shape=coordinate,label=right:$t_1$] at (tOne) {};
  \node[shape=coordinate,label=right:$t_2$] at (tTwo) {};
  \node[shape=coordinate,label=right:$t_2^*$] at (tTwoStar) {};
  
  \node[shape=coordinate,shift={(0.0,-0.2)},label=right:$c_{2,\false}^1$] at (cTwoFalseOne) {};
  \node[shape=coordinate,shift={(0.0,0.2)},label=right:$c_{2,\true}^1$] at (cTwoTrueOne) {};
  \node[shape=coordinate,label=below:$c_{2,\false}^0$] at (cTwoFalseZero) {};
  \node[shape=coordinate,label=above:$c_{2,\true}^0$] at (cTwoTrueZero) {};
  \node[shape=coordinate,shift={(-0.1,0)},label=left:$r_2$] at (rTwo) {};
  
  \node[shape=coordinate,shift={(0.0,0.2)},label=right:$c_{1,\false}^1$] at (cOneFalseOne) {};
  \node[shape=coordinate,shift={(0.0,-0.2)},label=right:$c_{1,\true}^1$] at (cOneTrueOne) {};
  \node[shape=coordinate,label=above:$c_{1,\false}^0$] at (cOneFalseZero) {};
  \node[shape=coordinate,label=below:$c_{1,\true}^0$] at (cOneTrueZero) {};
  \node[shape=coordinate,shift={(-0.05,0)},label=left:$r_1$] at (rOne) {};
  
\end{tikzpicture}
\end{center}
\vspace{-10pt}
\end{wrapfigure}

\paragraph{Construction of the curves} 
The curve $P_k$ will consist of all assignment gadgets for assignments $A_k$, $k \in \{1,2\}$, plus some additional points. The additional points implement an OR-gadget over the assignment gadgets, by enforcing that any traversal of $(P_1,P_2)$ with width $1+\eps$ has to traverse two assignment gadgets in parallel, and traversing one pair of assignment gadgets in parallel suffices.  

We define the following control points,
\begin{align*}
  &s_1 := (-\tfrac13,\tfrac15), \quad t_1 := (\tfrac13,\tfrac15), \\
  &s_2 := (-\tfrac13,0), \quad t_2 := (\tfrac13,0), \quad s_2^* := (-\tfrac13,-\tfrac45), \quad t_2^* := (\tfrac 13,-\tfrac45).
\end{align*}

Finally, we set 
\begin{align*}
  &P_1 := \bigcirc_{a_1 \in A_1} \big( s_1 \circ AG(a_1) \circ t_1 \big),  \\
  &P_2 := s_2 \circ s_2^* \circ \Big( \bigcirc_{a_2 \in A_2} AG(a_2) \Big) \circ t_2^* \circ t_2.
\end{align*}
The figure to the right shows $P_1$ (dotted) and $P_2$ (solid) in an example with $\mphi=2$ clauses and (unrealistically) only two assignments.

Let $Q_k$ be the vertices that may appear in $P_k$, i.e., $Q_1 = \{s_1,t_1,r_1,c_{1,\false}^0,c_{1,\true}^0,c_{1,\false}^1,c_{1,\true}^1\}$ and $Q_2 = \{s_2,t_2, r_2,s_2^*, t_2^*, c_{2,\false}^0,c_{2,\true}^0, c_{2,\false}^1,c_{2,\true}^1\}$.
Explicitly calculating all pairwise distances of all points, we obtain the following lemma.
\begin{lemma} \label{lem:distances}
No pair $(q_1,q_2) \in Q_1 \times Q_2$ has $\|q_1 - q_2\| \in (1,1+\eps]$. Moreover, the set $\{(q_1,q_2) \in Q_1 \times Q_2 \mid \|q_1 - q_2\| \le 1\}$ consists of the following pairs:
\begin{align*}
  &(q,s_2), (q,t_2) \text{ for any } q \in Q_1, \\
  &(s_1,q) \text{ for any } q \in Q_2 \setminus \{t_2^*\}, \\
  &(t_1,q) \text{ for any } q \in Q_2 \setminus \{s_2^*\}, \\
  &(r_1,r_2),  \\
  &(c_{1,x}^i,c_{2,y}^i) \text{ for } x \vee y = \true \text{ where } i \in \{0,1\}, x,y \in \{\true,\false\}.
\end{align*}
\end{lemma}

\paragraph{Correctness}
We show that if $\varphi$ is satisfiable then $\ddfr(P_1,P_2) \le 1$, while otherwise $\ddfr(P_1,P_2) > 1+\eps$. 

\begin{lemma} \label{lem:constIf}
  If $\ddfr(P_1,P_2) \le 1+\eps$ then $A_1 \times A_2$ contains a satisfying assignment.
\end{lemma}
\begin{proof}
  By \lemref{distances} any traversal with width $1+\eps$ also has width 1.
  Consider any traversal of $(P_1,P_2)$ with width 1. Consider any time step $T$ at which we are at position $s_2^*$ in $P_2$. The only point in~$P_1$ that is within distance $1$ of $s_2^*$ is $s_1$, say we are at the copy of $s_1$ that comes right before assignment gadget $AG(a_1)$, $a_1 \in A_1$. Following time step $T$, we have to start traversing $AG(a_1)$, so consider the first time step $T'$ where we are at the point $r_1$ in $AG(a_1)$. The only points in $P_2$ within distance $1$ of $r_1$ are $s_2, t_2,$ and~$r_2$. Note that we already passed $s_2^*$ in $P_2$ by time $T$, so we cannot be in~$s_2$ at time~$T'$. Moreover, in between $T$ and $T'$ we are only at $s_1$ and $r_1$ in $P_1$, which have distance larger than 1 to $t_2^*$. Thus, we cannot pass $t_2^*$, and we cannot be at $t_2$ at time $T'$. Hence, we are at $r_2$, say at the copy of $r_2$ in assignment gadget $AG(a_2)$ for some $a_2 \in A_2$. The yet untraversed remainder of $P_k$ is of the form $AG(a_k) \circ \pi_k$ for $k \in \{1,2\}$. Since our traversal of $(P_1,P_2)$ has width~1, we obtain $\ddfr(AG(a_1) \circ \pi_1, AG(a_2) \circ \pi_2) \le 1$. By \lemref{AG}, $(a_1,a_2)$ forms a satisfying assignment of~$\varphi$.
\end{proof}

\begin{lemma} \label{lem:constOnlyIf}
  If $A_1 \times A_2$ contains a satisfying assignment then $\ddfr(P_1,P_2) \le 1$.
\end{lemma}
\begin{proof}
  Let $(a_1,a_2) \in A_1 \times A_2$ be a satisfying assignment of $\varphi$. We describe a traversal through $P_1,P_2$ with width 1. We start at $s_2 \in P_2$ and the first point of $P_1$. We stay at $s_2$ and follow $P_1$ until we arrive at the copy of $s_1$ that comes right before $AG(a_1)$ (note that $s_2$ has distance 1 to any point in $P_1$). Then we stay at $s_1$ and follow $P_2$ until we arrive at the copy of $r_2$ in $AG(a_2)$ (note that the only point that is too far away from $s_1$ is $t_2^*$, but this point comes after all assignment gadgets in $P_2$). In the next step we go to positions $(r_1,r_2)$ (in $AG(a_1),AG(a_2)$). Then we follow the clause gadgets $(CG(a_1,i),CG(a_2,i))$ in parallel, always staying within distance~1 by \lemref{CG}. In the next step we stay at $CG(a_2,\mphi)$ and go to $t_1$ in $P_1$ (which has distance 1 to any point in $P_2$ except for $s_2^*$, which we will never encounter again). We stay at $t_1$ in $P_1$ and follow $P_2$ completely until we arrive at its endpoint~$t_2$. Since $t_2$ has distance 1 to any point in $P_1$, we can now stay at $t_2$ in~$P_2$ and follow $P_1$ to its end.
\end{proof}

\paragraph{Proof of \thmref{main}, discrete case}
Note that we have 
$$n = \max\{|P_1|,|P_2|\} = \Oh(\mphi) \cdot \max\{|A_1|,|A_2|\} = \Oh(\mphi\cdot 2^{\nphi/2}).$$
Moreover, the instance $(P_1,P_2)$ can be constructed in time $\Oh(\nphi \mphi 2^{\nphi/2})$. Any $(1+\eps)$-approximation can decide whether $\ddfr(P_1,P_2) \le 1$ or $\ddfr(P_1,P_2) > 1+\eps$, which by \lemrefs{constIf}{constOnlyIf} yields an algorithm that decides whether $\varphi$ is satisfiable. If such an algorithm runs in time $\Oh(n^{2-\delta})$ for any small $\delta > 0$, then the resulting \CNFSAT\ algorithm runs in time $\Oh(\mphi^2 2^{(1-\delta/2)\nphi})$, contradicting \CNFSETH.

\subsection{Continuous case} \label{sec:nondiscrete}

The construction from the last section also works for the continuous \Fr distance. However, for unsatisfiable formulas it becomes tedious to argue that continuous traversals are not much better than discrete traversals. For instance, we have to argue that we cannot stay at a fixed point between the clause gadgets $c_{1,\true}^0$ and $c_{1,\true}^1$ while traversing more than one clause gadget in~$P_2$.

We adapt the proof from the last section on the same curves $P_1,P_2$ to work for the continuous \Fr distance. To this end, we have to reprove \lemrefs{constIf}{constOnlyIf}. We will make use of the following property. Here, we set $\sym(CG(a_1,i)) := CG(a_2,i)$ and $\sym(r_1) := r_2$ and interpolate linearly between them to obtain a symmetric point in $AG(a_2)$ for every point in $AG(a_1)$ (for any fixed $a_1 \in A_1$, $a_2 \in A_2$). We also set $\sym(\sym(p_1)) := p_1$, to obtain a symmetric point in $AG(a_1)$ for every point in $AG(a_2)$.

\begin{lemma} \label{lem:nondiscAG}
  Consider any points $p_k$ in $AG(a_k)$, $k \in \{1,2\}$, with $\|p_1 - p_2\| \le 1+\eps$. Then we have $\|p_2 - \sym(p_1)\| \le \tfrac19$ and $\|\sym(p_2) - p_1\| \le \tfrac19$. 
\end{lemma}
\begin{proof}
  Let $p_k = (x_k,y_k)$ and note that we have $|y_1-y_2| \ge 1-2\eps$. Thus, if $|x_1-x_2| > \tfrac19-2\eps$ then we have (recall that $\eps = 1/1000$)
  $$ \|p_1 - p_2\| > \sqrt{ (\tfrac19-2\eps)^2 + (1-2\eps)^2 } > 1+\eps, $$
  a contradiction.
  Since $\sym(p_1) = (x_1,y_1')$ with $|y_1'-y_2| \le 2\eps$, we obtain
  $$ \|p_2 - \sym(p_1)\| \le \sqrt{(\tfrac19-2\eps)^2 + (2\eps)^2} \le \tfrac19.$$
  and the same bound holds for $\|\sym(p_2) - p_1\|$.
\end{proof}

\begin{lemma} \label{lem:constIfnd}
  \emph{(Analogue of \lemref{constIf})} If $\dfr(P_1,P_2) \le 1+\eps = 1.001$ then $A_1 \times A_2$ contains a satisfying assignment.
\end{lemma}
\begin{proof} 
  In this proof, we say that two points $p_1 = (x_1,y_1)$, $p_2 = (x_2,y_2)$ have \emph{$y$-distance} $D$ if $|y_1 - y_2| \le D$.
  
  Consider any traversal of $(P_1,P_2)$ with width $1+\eps$. Consider any time step $T$ where we are at position $s_2^*$ in $P_2$. The only points in $P_1$ that are within distance $1+\eps$ of $s_2^*$ are within distance $1/20$ and $y$-distance $\eps$ of $s_1$ (since no point in $P_1$ has lower $y$-value than $s_1$ and $\sqrt{1+(1/20)^2} > 1+\eps$). Say we are near the copy of $s_1$ that comes right before assignment gadget $AG(a_1)$, $a_1 \in A_1$. Following time step $T$, we have to start traversing $AG(a_1)$, so consider the first time step $T'$ where we are at the point $r_1$ in $AG(a_1)$. The only points in~$P_2$ within distance $1+\eps$ of $r_1$ are near $s_2, t_2,$ or~$r_2$. Note that we already passed $s_2^*$ in $P_2$ by time $T$, so we cannot be near $s_2$ at time~$T'$. Moreover, in between $T$ and $T'$ we are always near $s_1$ or between $s_1$ and $r_1$ in $P_1$, so we are always above and to the left of $s_1 + (1/20,0)$, which has distance larger than $1+\eps$ to~$t_2^*$. Thus, we cannot pass~$t_2^*$, and we cannot be near~$t_2$ at time $T'$. Hence, we are near $r_2$, more precisely, we are in distance 1/20 and $y$-distance $\eps$ of $r_2$ (this is the same situation as for $s_1$ and $s_2^*$).
  After that, the traversal has to further traverse $AG(a_1)$ and/or $AG(a_2)$. Consider the first time step at which we are at $CG(a_1,1)$ or $CG(a_2,1)$, say we reach $CG(a_1,1)$ first. By \lemref{nondiscAG}, we are within distance 1/9 of $CG(a_2,1)$. Since we were near~$r_2$ at time $T'$, we now passed $r_2$, and since we did not pass $CG(a_2,1)$ yet, we are even within distance 1/9 of $CG(a_2,1)$ \emph{along the curve $P_2$}. 
  This proves the induction base of the following inductive claim.
  
  
  \begin{claim}
    Let $T_i$ be the first step in time at which the traversal is at $CG(a_1,i)$ or $CG(a_2,i)$, $i \in [M]$. At time $T_i$ the traversal is within distance 1/9 of $CG(a_k,i)$ along the curve $P_k$ for both $k \in\{1,2\}$.
  \end{claim}
  \begin{proof}
    Note that at all times $T_i$ (and in between) \lemref{nondiscAG} is applicable, so we clearly are within distance 1/9 of $CG(a_k,i+1)$ at time $T_{i+1}$ for any $i \in [\mphi]$, $k\in\{1,2\}$. Since $\|CG(a_k,i) - CG(a_k,i+1)\| \ge 1/3$, points within distance 1/9 of $CG(a_k,i)$ are not within distance 1/9 of $CG(a_k,i+1)$. Hence, if we are within distance 1/9 of $CG(a_k,i)$ along $P_k$ for both $k \in \{1,2\}$ at time $T_i$, then at time $T_{i+1}$ we passed $CG(a_k,i)$ and did not pass $CG(a_k,i+1)$ yet (by definition of $T_{i+1}$), so that we are within distance 1/9 of $CG(a_k,i+1)$ along $P_k$ for both $k \in \{1,2\}$.
  \end{proof}
  
  Finally, we show that the above claim implies that $(a_1,a_2)$ is a satisfying assignment. Assume for the sake of contradiction that some clause $C_i$ is not satisfied by both $a_1$ and $a_2$. Say at time~$T_i$ we are at $CG(a_1,i)$ (if we are at $CG(a_2,i)$ instead, then a symmetric argument works). At the same time we are at some point $p$ in $AG(a_2)$. By the above claim, $p$ is within distance 1/9 of $CG(a_2,i)$ along $P_2$. Note that $p$ lies on any of the line segments $c_{2,\true}^0 \circ c_{2,\false}^1$, $c_{2,\false}^0 \circ c_{2,\true}^1$, $c_{2,\false}^0 \circ c_{2,\false}^1$, or $r_2 \circ c_{2,\false}^0$, since $\sat(a_{2},C_i) = \false$. In any case, the current distance $\|p - CG(a_1,i)\|$ is at least the distance from the point $c_{1,\false}^0$ to the line through $c_{2,\false}^0$ and $c_{2,\true}^1$. We compute this distance as
  $$ \frac{\tfrac13 (1+2\eps)}{\sqrt{(\tfrac13)^2 + (2\eps)^2}} > 1+\eps, $$
  which contradicts the traversal having width $1+\eps$.
\end{proof}

\begin{lemma} \label{lem:constOnlyIfnd}
  \emph{(Analogue of \lemref{constOnlyIf})} If $A_1 \times A_2$ contains a satisfying assignment then $\dfr(P_1,P_2) \le 1$.
\end{lemma}
\begin{proof}
  Follows from \lemref{constOnlyIf} and the general inequality
  $\dfr(P_1,P_2) \le \ddfr(P_1,P_2)$.
\end{proof}

\subsection{Generalization to imbalanced numbers of vertices} \label{sec:nm}

Assume that the input curves $P_1,P_2$ have different numbers of vertices $n=|P_1|$, $m=|P_2|$ with $n \ge m$. We show that there is no $\Oh((n m)^{1-\delta})$ algorithm for the \Fr distance for any $\delta > 0$, even for any polynomial restriction of $1 \le m \le n$. More precisely, for any $\delta \le \gamma \le 1-\delta$ we show that there is no $\Oh((nm)^{1-\delta})$ algorithm for the \Fr distance restricted to instances with $n^{\gamma-\delta} \le m \le n^{\gamma+\delta}$.

To this end, given a \CNFSAT\ instance $\varphi$ we partition its variables $x_1,\ldots,x_\nphi$ into\footnote{For the impatient reader: we will set $\ell := \nphi / (\gamma+1)$ (rounded in any way).} $V_1' := \{x_1,\ldots,x_\ell\}$ and $V_2' := \{x_{\ell+1},\ldots,x_\nphi\}$ and let $A_k'$ be all assignments of~$V_k'$, $k \in \{1,2\}$. Note that $|A_1'| = 2^{|V_1'|} = 2^\ell$ and $|A_2'| = 2^{\nphi-\ell}$. 
Now we use the same construction as in \secref{discrete} but replace $V_k$ by $V_k'$ and $A_k$ by $A_k'$. Again we obtain that any $1.001$-approximation for the \Fr distance of the constructed curves $P_1,P_2$ decides satisfiability of $\varphi$. Observe that the constructed curves contain a number of points of 
$$ n = |P_1| = \Theta(\mphi \cdot |A_1'|), \quad m = |P_2| = \Theta(\mphi \cdot |A_2'|). $$
Hence, any $1.001$-approximation of the \Fr distance with runtime $\Oh((nm)^{1-\delta})$ for any small $\delta>0$ yields an algorithm for \CNFSAT\ with runtime $\Oh(\mphi^2 (2^\ell 2^{\nphi-\ell})^{1-\delta}) = \Oh(\mphi^2 2^{(1-\delta)\nphi})$, contradicting \CNFSETH. 

Finally, we set $\ell := \nphi / (\gamma+1)$ (rounded in any way) so that $|A_1'| = \Theta(2^{\nphi/(\gamma+1)})$ and $|A_2'| = \Theta(2^{\nphi \gamma/(\gamma+1)})$. Using \lemref{seth} we can assume that $1 \le \mphi \le 2^{\delta \nphi / 4}$. Hence, we have 
\begin{align*}
  \Omega( 2^{\nphi/(\gamma+1)} ) \le \,&n \le \Oh(2^{\nphi/(\gamma+1) + \delta\nphi/4}),  \\
  \Omega( 2^{\nphi \gamma/(\gamma+1)} ) \le \,&m \le \Oh(2^{\nphi \gamma/(\gamma+1) + \delta\nphi/4}),
\end{align*}
which implies
$\Omega(n^{\gamma - \delta/2}) \le m \le \Oh(n^{\gamma + \delta/2})$.
For sufficiently large $n$, we obtain the desired polynomial restriction $n^{\gamma - \delta} \le m \le n^{\gamma + \delta}$. This proves \thmref{const}.

\section{Realistic inputs: $c$-packed curves} \label{sec:cpacked}

\subsection{Constant factor approximations} \label{sec:cpackedconst}

The curves constructed in \secref{discrete} are highly packed, since all assignment gadgets lie roughly in the same area. Specifically they are not $o(n)$-packed. In this section we want to construct $c$-packed instances and show that there is no $1.001$-approximation with runtime $\Oh((cn)^{1-\delta})$ for any $\delta>0$ for the \Fr distance unless \CNFSETH\ fails, not even restricted to instances with $n^{\gamma-\delta} \le c \le n^{\gamma+\delta}$ for any $\delta \le \gamma \le 1-\delta$.
This proves \thmref{cpacked}.

To this end, we again consider a \CNFSAT\ instance $\varphi$, partition its variables $x_1,\ldots,x_\nphi$ into two sets $V_1,V_2$ of size $N/2$, and consider the set $A_k$ of all assignments of \true\ and \false\ to the variables in~$V_k$. Now we partition $A_k$ into sets $A_k^1,\ldots,A_k^\ell$ of size $\Theta(2^{\nphi/2}/\ell)$, where we fix $1 \le \ell \le 2^{\nphi/2}$ later. Formula $\varphi$ is satisfiable if and only if for some pair $(j_1,j_2) \in [\ell]^2$ the set $A_1^{j_1} \times A_2^{j_2}$ contains a satisfying assignment. 
This suggests to use the construction of \secref{discrete} after replacing $A_1$ by $A_1^{j_1}$ and $A_2$ by $A_2^{j_2}$, yielding a pair of curves $(P_1^{j_1j_2},P_2^{j_1j_2})$. Now, $\varphi$ is satisfiable if and only if $\dfr(P_1^{j_1j_2},P_2^{j_1j_2}) \le 1$ for some $(j_1,j_2) \in [\ell]^2$. For the sake of readability, we rename the constructed curves slightly so that we have curves $(P_1^j,P_2^j)$ for $j \in [\ell^2]$.

\begin{wrapfigure}{r}{0.16\textwidth}
\vspace{-0pt}
\begin{center}
\begin{tikzpicture}[every node/.style={fill, circle, inner sep = 1pt},scale=1.7]
  \drawUL;
  \drawUR;
  
  \defpts;
  
  \node[shape=coordinate,label={[label distance=0.07cm]right:$U_L$}] at (UL4) {};
  \node[shape=coordinate,label={[label distance=0.07cm]left:$U_R$}] at (UR2) {};
  
  \path[dotted] (UL5) edge node[above,style={fill=white}] {$U$} (UR1);
\end{tikzpicture}
\end{center}
\vspace{-10pt}
\end{wrapfigure}

\paragraph{OR-gadget} 
In the whole section we let $\rho := 1/\sqrt{2}$.
We present an OR-construction over the gadgets $(P_1^j,P_2^j)$ that is not too packed, in contrast to the OR-construction over assignment gadgets that we used in \secref{discrete}. 
We start with two building blocks, where for any $j \in \N$ we set
\begin{align*}
  U_L(j) &:= (j\rho,0) \circ ((j-1)\rho,\rho) \circ ((j-1)\rho,3\rho) \circ ((j-1)\rho,2\rho) \circ ((j-1)\rho,\rho), \\
  U_R(j) &:= ((j+1)\rho,\rho) \circ ((j+1)\rho,2\rho) \circ ((j+1)\rho,3\rho) \circ ((j+1)\rho,\rho) \circ (j\rho,0).
\end{align*}

Moreover, we set $U(j) := U_L(j) \circ U_R(j)$.
For a curve $\pi$ and $z \in \R$ we let $\Tr_z(\pi)$ be the curve~$\pi$ translated by $z$ in $x$-direction.
The OR-gadget now consists of the following two curves,
\begin{align*}
  R_1 &:= \bigcirc_{j=1}^{\ell^2} \big( U_L(2j) \circ \Tr_{2j\rho}(P_1^{j}) \circ U_R(2j)\big),  \\
  R_2 &:= U(1) \circ \bigcirc_{j=1}^{\ell^2} \big(  \Tr_{2j\rho}(P_2^{j}) \circ U(2j+1) \big).
\end{align*}

\begin{wrapfigure}{r}{0.4\textwidth}
\vspace{-10pt}
\begin{center}
\begin{tikzpicture}[every node/.style={fill, circle, inner sep = 0.5pt},scale=0.85]
  
  \foreach \j in {0,2,4,6}{
  \begin{scope}[xshift=\j*\myrho cm,every path/.style={dotted}]
    \drawUL;
    \drawPOne;
    \drawUR;
    
    \defpts;
    \draw (UL5) -- (sOne);
    \draw (tOne) -- (UR1);
  \end{scope}
  }
  \draw[dotted] (0,0) -- (6*\myrho,0);
  
  \foreach \j in {0,2,4,6}{
  \begin{scope}[xshift=\j*\myrho cm]
    \begin{scope}[xshift=-\myrho cm]
      \drawU;
    \end{scope}
    
    \defpts
    \draw (-\myrho,0) -- (sTwo);
    \draw (tTwo) -- (\myrho,0);
    
    \drawPTwo
  
    \begin{scope}[xshift=\myrho cm]
      \drawU;
    \end{scope}
  \end{scope}
  }
\end{tikzpicture}
\end{center}
\vspace{-10pt}
\end{wrapfigure}

\noindent
The figure to the right shows $R_1$ (dotted) and $R_2$ (solid) for $\ell^2=4$, see below for a figure showing $\ell^2=1$ with more details visible.

We denote by $R_1^j$ the $j$-th ``summand'' of $R_1$, i.e., $R_1^j = U_L(2j) \circ \Tr_{2j\rho}(P_1^{j}) \circ U_R(2j)$. Informally, we will use the term \emph{$U$-shape} for the subcurves $R_1^j$ and $U(2j+1)$, since they resemble the letter~U. Moreover, we consider ``summands'' of $R_2$, namely $R_2^j := U(2j-1) \circ \Tr_{2j\rho}(P_2^{j}) \circ ((2j+1)\rho,0)$ and $\tilde R_2^j := ((2j-1)\rho,0) \circ \Tr_{2j\rho}(P_2^{j}) \circ U(2j+1)$.

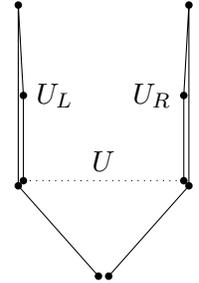
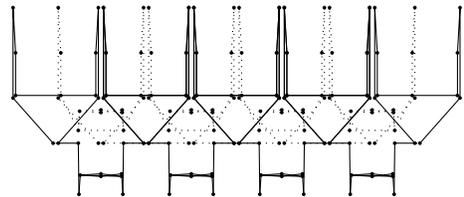
\begin{wrapfigure}{r}{0.32\textwidth}
\vspace{-2pt}
\begin{center}
\begin{tikzpicture}[every node/.style={fill, circle, inner sep = 1pt},scale=1.7]
  
  \begin{scope}[every path/.style={dotted}]
    \drawUL;
    \drawPOne;
    \drawUR;
    
    \defpts;
    \draw (UL5) -- (sOne);
    \draw (tOne) -- (UR1);
  \end{scope}
  
  \begin{scope}
    \begin{scope}[xshift=-\myrho cm]
      \drawU;
    \end{scope}
    
    \defpts
    \draw (-\myrho,0) -- (sTwo);
    \draw (tTwo) -- (\myrho,0);
    
    \drawPTwo
  
    \begin{scope}[xshift=\myrho cm]
      \drawU;
    \end{scope}
  \end{scope}
\end{tikzpicture}
\end{center}
\vspace{-20pt}
\end{wrapfigure}

\paragraph{Intuition}
Considering traversals that stay within distance 1, we can traverse one $U$-shape in~$R_1$ and one neighboring $U$-shape in $R_2$ together. Such traversals can be stitched together to a traversal of any number $j$ of neighboring $U$-shapes in both curves. So far we can only traverse the same number of $U$-shapes in both curves, but $R_2$ has one more $U$-shape than $R_1$. We will show that we can traverse two $U$-shapes in $R_2$ while traversing only one $U$-shape in $R_1$, if these parts contain a satisfying assignment.

In the unsatisfiable case, essentially we show that we cannot traverse two $U$-shapes in $R_2$ while traversing only one $U$-shape in $R_1$, which implies a contradiction since the number of $U$-shapes in $R_2$ is larger than in $R_1$. We make this intuition formal in the remainder of this section.

\paragraph{Analysis}
In order to be able to replace the curves $P_1^j,P_2^j$ constructed above by other curves in the next section, we analyse the OR-gadget in a rather general way. To this end, we first specify a set of properties and show that the curves $P_1^j,P_2^j$ constructed above satisfy these properties. Then we analyse the OR-gadget using only these properties of $P_1^j,P_2^j$.

\begin{property} \label{prop:PG}
  \begin{enumerate}[(i)] \itemsep0em 
    \item If $\varphi$ is satisfiable then for some $j \in [\ell^2]$ we have $\ddfr(P_1^j,P_2^j) \le 1$.
    \item If $\varphi$ is not satisfiable then for all $j \in [\ell^2]$ and curves $\sigma_1,\sigma_2,\pi_1,\pi_2$ such that $\sigma_1$ stays to the left and above $(-\rho,\rho)$ and $\pi_1$ stays to the right and above $(\rho,\rho)$,  we have $\dfr(\sigma_1 \circ P_1^j \circ \pi_1, \sigma_2 \circ P_2^j \circ \pi_2) > \beta$, for some $\beta > 1$.
    \item $P_k^j$ is $\Theta(c)$-packed for some $c\ge1$ for all $j \in [\ell^2]$, $k \in \{1,2\}$.
    \item $(0,\rho)$ is within distance 1 of any point in $P_1^j$ for all $j \in [\ell^2]$.
    \item $(0,0)$ is within distance 1 of any point in $P_2^j$ for all $j \in [\ell^2]$.
  \end{enumerate}
\end{property}

\begin{lemma} \label{lem:propConst}
  The curves $(P_1^j,P_2^j)$ constructed above satisfy \propref{PG} with $\beta = 1.001$ and $c = \Theta(\mphi \cdot 2^{\nphi/2}/\ell)$. Moreover, we have $|P_k^j| = \Theta(\mphi \cdot 2^{\nphi/2}/\ell)$ for all $j \in [\ell^2]$, $k \in \{1,2\}$.
\end{lemma}
\begin{proof}
  \propref{PG}.(i) follows from \lemref{constOnlyIf}, since at least one pair $(A_1^{j_1},A_2^{j_2})$ contains a satisfying assignment. 
  Properties (iv) and (v) can be verified by considering all points in the construction in \secref{discrete}. 
  
  Observe that $|P_k^j| = \Theta(\mphi \cdot 2^{\nphi/2} / \ell)$, since $P_k^j$ consists of $|A_k^j| = \Theta( 2^{\nphi/2} / \ell)$ assignment gadgets of size $\Theta(\mphi)$. The upper bound of (iii) follows since any polygonal curve with at most $m$ segments is $m$-packed. The lower bound of (iii) follows from $P_k^j$ being contained in a ball of radius 1 (by (iv) and (v)) and every segment of $P_k^j$ having constant length.
  
  For (ii), note that from any traversal of $(\sigma_1 \circ P_1^j \circ \pi_1, \sigma_2 \circ P_2^j \circ \pi_2)$ with width $1.001$ one can extract a traversal of $(P_1^j,P_2^j)$ with width $1.001$, by mapping any point in $\sigma_k$ to the starting point~$s_k$ of $P_k^j$ and any point in $\pi_k$ to the endpoint $t_k$ of $P_k^j$, $k \in \{1,2\}$. This does not increase the width, since (1) $s_2$ and $t_2$ are within distance~1 to all points in $P_1^j$, 
  and (2) $s_1$ has smaller distance to any point in $P_2^j$ than any point in $\sigma_1$ has, since $\sigma_1$ stays above and to the left of~$s_1$ while all points of $P_1^j$ lie below and to the right of $s_1$. A similar statement holds for $t_1$ and $\pi_1$. Property (ii) now follows from \lemref{constIfnd}.
\end{proof}

In the following lemma we analyse the OR-gadget.

\begin{lemma} \label{lem:OR}
  For any curves $(P_1^j,P_2^j)$ that satisfy \propref{PG}, the OR-gadget $(R_1,R_2)$ satisfies:
  \begin{enumerate}[(i)] \itemsep0em 
    \item $|R_k| = \Theta\big(\sum_{j=1}^{\ell^2} |P_k^j|\big)$ for $k \in \{1,2\}$.
    \item $R_1$ and $R_2$ are $\Theta(c)$-packed,
    \item If $\varphi$ is satisfiable then $\dfr(R_1,R_2) \le \ddfr(R_1,R_2) \le 1$,
    \item If $\varphi$ is not satisfiable then $\ddfr(R_1,R_2) \ge \dfr(R_1,R_2) > \min\{\beta,1.2\}$.
  \end{enumerate}
\end{lemma}
\begin{proof}
  (i) Precisely, we have $|R_k| = \sum_{j=1}^{\ell^2} (|P_k^j| + 10) + 10(k-1)$ for $k \in \{1,2\}$. 
  \vspace{0.9em}
  
  (ii) Let $k \in \{1,2\}$ and consider any ball $B = B(q,r)$. If $r \le 1$ then $B$ hits $\Oh(1)$ of the curves $P_k^j$. Since these curves are $c$-packed, their contribution to the total length of $R_k$ in~$B$ is at most $\Oh(c r)$. Moreover, $B$ hits $\Oh(1)$ segments of $U$ or $U_L,U_R$, and the connecting segments to $P_k^j$. Each of these segments has length at most $2 r$ inside $B$. This yields a total length of $R_k$ in $B$ of $\Oh((c+1) r)$.
  
  Similarly, if $r > 1$ then $B$ hits $\Oh(r)$ of the curves $P_k^j$. Note that the total length of $P_k^j$ is at most $c$, since the curve is $c$-packed and contained in a ball of radius~1 around $(0,0)$ or $(0,\rho)$ by \propref{PG}. Hence, the total length of of the curves $P_k^j$ in $B$ is $\Oh(cr)$. Moreover, $B$ hits $\Oh(r)$ segments of $U,U_L,U_R$, and the connectors to $P_k^j$, each of constant length. This yields a total length of $R_k$ in $B$ of $\Oh((c+1) r)$.
  
  In total, the curve $R_k$ is $\Oh(c+1)$-packed. 
  As $c \ge 1$, it is also $\Oh(c)$-packed. 
  Since for some $\alpha>0$ the curve $P_k^j$ is not $\alpha c$-packed, also $R_k$ is not $\alpha c$-packed, so $R_k$ is even $\Theta(c)$-packed.
  \vspace{0.9em}
  
  (iii) Note that $\dfr(R_1,R_2) \le \ddfr(R_1,R_2)$ holds in general, so we only have to show that if $\varphi$ is satisfiable then $\ddfr(R_1,R_2) \le 1$. 
  First we show that we can traverse one $U$-shape in~$R_1$ and one neighboring $U$-shape in $R_2$ together.
  
  \begin{claim}
    For any $j \in [\ell^2]$, we have $\ddfr(R_1^j, U(2j-1)) \le 1$ and $\ddfr(R_1^j, U(2j+1)) \le 1$.
  \end{claim}
  \begin{proof}
    We only show the first inequality, the second is similar. We start by traversing $U_L(2j)$ and the left half of $U(2j-1)$ in parallel, being at the $i$-th point of $U_L(2j)$ and $U(2j-1)$ at the same time. At any point in time we are within distance $\rho$. Now we step to $(2j\rho,\rho)$ in $U(2j-1)$. We stay there while traversing $\Tr_{2j\rho}(P_1^{j})$ in $R_1^j$, staying within distance~1 by \propref{PG}.(iv). Finally, we traverse $U_R(2j)$ and the second half of $U(2j-1)$ in parallel, where again the largest encountered distance is $\rho$.
  \end{proof}
  We can stitch these traversals together so that we traverse any number $j$ of neighboring $U$-shapes in both curves together, because the parts in between the $U$-shapes are near to a single point, as shown by the following claim. Note that $(2j\rho,0) \circ ((2j+2)\rho,0)$ is the connecting segment in $R_1$ between $U_R(2j)$ and $U_L(2j+2)$, while $((2j-1)\rho,0) \circ \Tr_{2j\rho}(P_2^{j}) \circ ((2j+1)\rho,0)$ is the part in~$R_2$ between $U(2j-1)$ and $U(2j+1)$.
  \begin{claim}
    For any $j \in [\ell^2]$,
    \begin{align*}
      &\ddfr((2j\rho,0) \circ ((2j+2)\rho,0), ((2j+1)\rho,0)) \le 1,  \\
      &\ddfr((2j\rho,0), ((2j-1)\rho,0) \circ \Tr_{2j\rho}(P_2^{j}) \circ ((2j+1)\rho,0)) \le 1.
    \end{align*}
  \end{claim}
  \begin{proof}
    The first claim is immediate. The second follows from \propref{PG}.(v).
  \end{proof}
  Thus, we can stitch together traversals of $U$-shapes in both curves. However, so far we can only traverse the same number of $U$-shapes in both curves, but $R_2$ has one more $U$-shape than $R_1$. 
  Consider $J \in [\ell^2]$ with $\ddfr(P_1^J,P_2^J) \le 1$, which exists since $\varphi$ is satisfiable, see \propref{PG}.(i). 
  Consider the two subcurves (also see the above figure)
  \begin{align*}
    R_1' &:= R_1^J = U_L(2J) \circ \Tr_{2J\rho}(P_1^{J}) \circ U_R(2J),  \\
    R_2' &:= U(2J-1) \circ \Tr_{2J\rho}(P_2^{J}) \circ U(2J+1).
  \end{align*}
  We show that $\ddfr(R_1',R_2') \le 1$, i.e., we can traverse two $U$-shapes in $R_2$ while traversing only one $U$-shape in $R_1$, using $\ddfr(P_1^J,P_2^J) \le 1$. Adding simple traversals of $U$-shapes before and after $(R_1',R_2')$, we obtain a traversal of $(R_1,R_2)$ with width $1$, proving $\ddfr(R_1,R_2) \le 1$. It is left to show the following claim.
  \begin{claim}
    $\ddfr(R_1',R_2') \le 1$.
  \end{claim}
  \begin{proof}
    We traverse $U_L(2J)$ and $U(2J-1)$ in parallel until we are at point $((2J-1)\rho,2\rho)$ in $U_L(2J)$. We stay in this point and follow $U(2J-1)$ until its second-to-last point. In the next step we can finish traversing $U_L(2J)$ and $U(2J-1)$. In the next step we go to the first positions of (the translated) $P_1^J$ and $P_2^J$. We follow any traversal of $(P_1^J,P_2^J)$ with width 1. Finally, we use a traversal symmetric to the one of $(U_L(2J),U(2J-1))$ to traverse $(U_R(2J),U(2J+1))$. 
  \end{proof}
  \vspace{0.5em}
  
  (iv) Note that the inequality $\ddfr(R_1,R_2) \ge \dfr(R_1,R_2)$ holds in general, so we only have to show that if $\varphi$ is not satisfiable then $\dfr(R_1,R_2) > \min\{\beta,1.2\}$. Assume for the sake of contradiction that there is a traversal of $(R_1,R_2)$ with width $\min\{\beta,1.2\}$. Essentially we show that it cannot traverse 2 $U$-shapes in $R_2$ while traversing only one $U$-shape in $R_1$, which implies a contradiction since the number of $U$-shapes in $R_2$ is larger than in $R_1$.
  
  Let $Y_\rho$ be the line $\{(x,y) \in \R^2 \mid y=\rho\}$.
  We inductively prove the following claims.
  
  \begin{claim}
    \begin{enumerate}
      \item[(i)] For any $0 \le j \le \ell^2$, when the traversal is in $R_2$ at the left highest point $(2j\rho,3\rho)$ of $U(2j+1)$, then in $R_1$ we fully traversed $R_1^j$ and are above the line~$Y_\rho$.
      \item[(ii)] For any $1 \le j \le \ell^2$, when the traversal is in $R_1$ at the right highest point $((2j+1)\rho,3\rho)$ of $R_1^j$, then in $R_2$ it is in $U(2j-1)$.
    \end{enumerate}
  \end{claim}
  
  Note that claim (i) for $j=\ell^2$ yields the desired contradiction, since after traversing $R_1^{\ell^2}$ the curve $R_1$ has ended (at the point $(2 \ell^2 \rho,0)$), so that we cannot go above the line $Y_\rho$ anymore.
  
  \begin{proof}
    (i) Note that we have to be above the line $Y_\rho$ because all points below $Y_\rho$ have distance at least $2 \rho >1.2$ to the point $(2j\rho,3\rho)$.
    For $j=0$, claim (i) holds immediately, since there is no subcurve $R_1^0$ (so this part of the statement disappears). In general, claim (i) for any $1 \le j \le \ell^2$ follows from claim (ii) for $j$: 
    When we are at $z_1 := ((2j+1)\rho,3\rho)$ in $R_1^j$, we are still in $U(2j-1)$.
    Once we reach the endpoint $z_2 := ((2j-1)\rho,0)$ of $U(2j-1)$, in $R_1$ we are at a point $p_1$ below the line $Y_\rho$, since all points in $R_2$ that follow $z_1$ and lie above $Y_\rho$ have distance more than $2\rho > 1.2$ to $z_2$. Now we follow $R_2$ until we reach $p_2 := (2j\rho,3\rho)$ in $U(2j+1)$. At this point we have to be above the line $Y_\rho$ in $R_1$, but all points in $R_1^j$ following~$p_1$ lie below~$Y_\rho$. Thus, at this point we have fully traversed $R_1^j$ (and have to be in~$R_1^{j+1}$).
    
    (ii) This claim for any $1 \le j \le \ell^2$ follows from claim (i) for $j-1$. Assume for the sake of contradiction that claim (ii) for some $j$ does not hold. Consider the subcurve $R_1'$ of $R_1^j$ between (the first occurrence of) $((2j-1)\rho,\rho)$ and $((2j+1)\rho,3\rho)$. Let $R_2'$ be the subcurve of $R_2$ that the traversal traverses together with~$R_1'$. Since $(R_1',R_2')$ forms a subtraversal of the traversal of $(R_1,R_2)$, which has width $\min\{\beta,1.2\}$, we have $\dfr(R_1',R_2') \le \min\{\beta,1.2\}$ (*). By claim (i) for $j-1$, the starting point of $R_2'$ lies before $\Tr_{2j\rho}(P_2^j)$ along $R_2$, since we reach $((2j-2)\rho,3\rho)$ in $U(2j-1)$ only after being in the starting point of $R_1'$. Moreover, the endpoint of $R_2'$ lies after $\Tr_{2j\rho}(P_2^j)$ along $R_2$. Indeed, while being at the endpoint $((2j+1)\rho,3\rho)$ of $R_1'$, we cannot be in $U(2j-1)$ since we assumed that claim (ii) is wrong for $j$. We can also not be in $\Tr_{2j\rho}(P_2^j)$, since by \propref{PG}.5 all points in this curve lie in a ball of radius 1 around $(2j\rho,0)$, so their distance to $((2j+1)\rho,3\rho)$ is at least $\|((2j+1)\rho,3\rho) - (2j\rho,0)\|-1 = \sqrt{5}-1 > 1.2$. Hence, we already passed $\Tr_{2j\rho}(P_2^j)$, and $R_2'$ is of the form $\sigma_2 \circ \Tr_{2j\rho}(P_2^j) \circ \pi_2$ for any curves $\sigma_2, \pi_2$. Note that $R_1'$ is of the form $\sigma_1 \circ \Tr_{2j\rho}(P_1^j) \circ \pi_1$ with $\sigma_1$ staying above and to the left of $((2j-1)\rho,\rho)$ and $\pi_1$ staying above and to the right of $((2j+1)\rho,\rho)$. Thus, after translation \propref{PG}.(ii) applies, proving $\dfr(R_1',R_2') > \beta$, a contradiction to (*).
  \end{proof}
\end{proof}

\paragraph{Proof of \thmref{cpacked}} 
Finally, we use the OR-gadget (\lemref{OR}) together with the curves $P_1^j,P_2^j$ we obtained from \secref{discrete} (\lemref{propConst}) to prove a runtime bound for $c$-packed curves: Any $1.001$-approximation for the (discrete or continuous) \Fr distance of $(R_1,R_2)$ decides satisfiability of $\varphi$. Note that $R_1$ and $R_2$ are $c$-packed with 
$$ c = \Theta(\mphi \cdot 2^{\nphi/2} / \ell), \quad 
n = \max\{|R_1|,|R_2|\} = \Theta(\ell^2 \mphi \cdot 2^{\nphi/2} / \ell). $$
Thus, any $\Oh((cn)^{1-\delta})$ algorithm for the \Fr distance implies a 
$\Oh(\mphi^2 2^{(1-\delta)\nphi})$ algorithm for \CNFSAT, contradicting \CNFSETH. Moreover, using \lemref{seth} we can assume that $1 \le \mphi \le 2^{\delta \nphi/4}$. Setting $\ell := \Theta(2^{\frac{1-\gamma}{1+\gamma} \nphi/2})$ for any $0 \le \gamma \le 1$ we obtain
\begin{align*} 
  \Omega( 2^{\frac{2}{1+\gamma} \nphi/2} ) \le \, &n \le \Oh( 2^{(\frac{2}{1+\gamma} + \delta/2) \nphi/2} ),  \\
  \Omega( 2^{\frac{2\gamma}{1+\gamma} \nphi/2} ) \le \, &c \le \Oh( 2^{(\frac{2\gamma}{1+\gamma} + \delta/2) \nphi/2} ).
\end{align*}
From this it follows that $\Omega(n^{\gamma-\delta/2}) \le c \le \Oh(n^{\gamma+\delta/2})$, which implies the desired polynomial restriction $n^{\gamma-\delta} \le c \le n^{\gamma+\delta}$ for sufficiently large $n$.

\subsection{Approximation schemes} \label{sec:cpackedFive}

In this section, we consider the dependence on $\eps$ of the runtime of a $(1+\eps)$-approximation for the \Fr distance on $c$-packed curves. We show that in $\R^d$ with $d \ge 5$ there is no such algorithm with runtime $\Oh(\min\{cn/\sqrt{\eps}, n^2\}^{1-\delta})$ for any $\delta>0$ unless \CNFSETH\ fails (\thmref{cpackedFive}). This matches the dependence on $\eps$ of the fastest known algorithm up to a polynomial. The result holds for sufficiently small $\eps>0$ and any polynomial restriction of $1 \le c \le n$ and $\eps \le 1$.

We will reuse the OR-gadget from the last section, embedded into the first two dimensions of~$\R^5$. Specifically, we will reuse \lemref{OR}. However, we adapt the curves $P_1^j,P_2^j$, essentially by embedding the same set of points in a different way. 

\begin{wrapfigure}{r}{0.3\textwidth}
\vspace{-10pt}
\begin{center}
\begin{tikzpicture}[every node/.style={fill, circle, inner sep = 1pt},scale=1.2,cross/.style={fill=white,inner sep = 2pt,path picture={ 
  \draw[black]
(path picture bounding box.south east) -- (path picture bounding box.north west) (path picture bounding box.south west) -- (path picture bounding box.north east);
}}]
  \foreach \x in {0,...,10}{
    \node at (45+90*\x/10:1) {};
  }
  \node[cross] at (0,0) {};
  
  
  \begin{scope}[xshift=2cm]
  \foreach \x in {0,...,10}{
    \node[cross] at (45+90*\x/10:1) {};
  }
  \node at (0,0) {};
  \end{scope}
 
  \draw (1,1.3) -- (1,-0.4);
\end{tikzpicture}
\end{center}
\vspace{-10pt}
\end{wrapfigure}

In this new embedding we make use of the fact that in $\R^4$ there are point sets $Z_1,Z_2$ of arbitrary size such that any pair of points $(z_1,z_2) \in Z_1 \times Z_2$ has distance 1. For an example, see the figure to the right,  where the left picture shows the projection onto the first two dimensions and the right picture shows the projection onto the last two dimensions. Here, $Z_1$ (circles) is placed along a quarter-circle in the $(1,2)$-plane and $Z_2$ (crosses) is placed along a quarter-circle in the $(3,4)$-plane. 

\paragraph{Construction}
As usual, consider a \CNFSAT\ instance $\varphi$, partition its variables $x_1,\ldots,x_\nphi$ into two sets $V_1,V_2$ of size $N/2$, and consider any set $A_k$ of assignments of \true\ and \false\ to the variables in~$V_k$. Fix any enumeration $\{a_k^1,\ldots,a_k^{|A_k|}\}$ of $A_k$. Again set $\rho := 1/\sqrt{2}$. For $h \in [|A_k|]$ and $i \in \{0,\ldots,\mphi+1\}$ let
\begin{align*}
  \rot(a_1^h,i) &:= \big(\rho \sin\big(\tfrac\pi 4 + \tfrac\pi 2 \tfrac{h (\mphi +2) + i}{|A_1| \cdot (\mphi +2) }\big), \rho \cos\big(\tfrac\pi 4 + \tfrac\pi 2 \tfrac{h (\mphi +2)  + i}{|A_1| \cdot (\mphi +2) }\big),0,0,0\big),  \\
  \rot(a_2^h,i) &:= \big(0,0,\rho \sin\big(\tfrac\pi 4 + \tfrac\pi 2 \tfrac{h (\mphi +2)  + i}{|A_2| \cdot (\mphi +2) }\big), \rho \cos\big(\tfrac\pi 4 + \tfrac\pi 2 \tfrac{h (\mphi +2)  + i}{|A_2| \cdot (\mphi +2) }\big),0\big).
\end{align*}
Note that these points are placed along a quarter-circle in the $(1,2)$-plane or $(3,4)$-plane, respectively, as in the above figure. In particular, $\| \rot(a_1^h,i) - \rot(a_2^{h'},i')\| = 1$ for all $h,h',i,i'$.
Moreover, let $e_5$ be the vector $(0,0,0,0,\rho)$. For $a_k \in A_k$ and $i \in [\mphi]$ we set 
\begin{align*}
  CG(a_k,i) := \begin{cases} (1 - 2\eps) \,\rot(a_k,i) + (i \mod 2) \cdot 8 \sqrt{\eps} e_5,  &\text{if } \sat(a_k,C_i) = \true  \\
  (1 + \eps) \,\rot(a_k,i) + (i \mod 2) \cdot 8 \sqrt{\eps} e_5,  &\text{if } \sat(a_k,C_i) = \false \end{cases}
\end{align*}
Thus, we align the clause gadgets of $A_1$ roughly along a quarter-circle in the $(1,2)$-plane, and similarly the clause gadgets of $A_2$ roughly along a quarter-circle in the $(3,4)$-plane.
Moreover, for $a_k \in A_k$ we set
\begin{align*}
  &r_k(a_k) := \rot(a_k,0) - 8 \sqrt{\eps} e_5,  \\
  &s_1(a_1) := (1 - 400 \eps) \, \rot(a_1,0) + 10 \sqrt \eps e_5,  \\
  &t_1(a_1) := (1 - 400 \eps) \, \rot(a_1,\mphi+1) - 10 \sqrt \eps e_5,  \\
  &s_2 = t_2 := (0,0,0,0,0),  \\
  &s_2^* := (1 + 9 \sqrt \eps) e_5, \quad t_2^* := -(1+ 9 \sqrt \eps) e_5.
\end{align*}
We define assignment gadgets and the curves $P_1,P_2$ as in \secref{discrete}, i.e.,
\begin{align*}
  AG(a_k) &:= r_k(a_k) \circ \bigcirc_{i \in [\mphi]} CG(a_k,i),  \\
  P_1 &:= \bigcirc_{a_1 \in A_1} \big( s_1(a_1) \circ AG(a_1) \circ t_1(a_1) \big),  \\
  P_2 &:= s_2 \circ s_2^* \circ \Big( \bigcirc_{a_2 \in A_2} AG(a_2) \Big) \circ t_2^* \circ t_2.
\end{align*}

\paragraph{Analysis}
Again, we split the considered points into $Q_1,Q_2$, depending on whether they may appear on $P_1$ or $P_2$, i.e., $Q_1 := \{s_1(a_1), t_1(a_1), r_1(a_1), CG(a_1,i) \mid a_1 \in A_1, i \in [\mphi] \}$ and $Q_2 := \{s_2, t_2, s_2^*, t_2^*, r_2(a_2), CG(a_2,i) \mid a_2 \in A_2, i \in [\mphi] \}$. It is easy, but tedious to verify that the constructed points behave as follows. 

\begin{lemma} \label{lem:last}
  The following pairs of points have distance at most 1 for any $a_k \in A_k$:
  \begin{align*}
    &(q,s_2), (q,t_2) \text{ for any } q \in Q_1,  \\
    &(s_1(a_1),q) \text{ for any } q \in Q_2 \setminus \{t_2^*\},  \\
    &(t_1(a_1),q) \text{ for any } q \in Q_2 \setminus \{s_2^*\},  \\
    &(r_1(a_1),r_2(a_2)),  \\
    &(CG(a_1,i),CG(a_2,i)) \text{ if assignment $(a_1,a_2)$ satisfies clause $C_i$}.
  \end{align*}
  Moreover, the following pairs of points have distance more than $1+\eps$ for any $a_k \in A_k$:
  \begin{align*}
    &(q,s_2^*) \text{ for any } q \in Q_1 \setminus \{ s_1 \},  \\
    &(q,t_2^*) \text{ for any } q \in Q_1 \setminus \{ t_1 \},  \\
    &(r_1(a_1), CG(a_2,i)) \text{ for any } i \in [\mphi],  \\
    &(CG(a_1,i), r_2(a_2)) \text{ for any } i \in [\mphi],  \\
    &(CG(a_1,i), CG(a_2,j)) \text{ for any } i,j \in [\mphi], i \not\equiv j \bmod 2,  \\
    &(CG(a_1,i),CG(a_2,i)) \text{ if assignment $(a_1,a_2)$ does not satisfy clause $C_i$}.
  \end{align*}
\end{lemma}
\begin{proof}
  Using that $\eps$ is sufficiently small, we only have to compute the largest order term of~$\eps$ for all distances. E.g., for all $a_k \in A_k$
  $$\|s_1(a_1) - r_2(a_2)\| = \sqrt{ \rho^2( (1-400 \eps)^2 + 1 + (18 \sqrt \eps)^2) } = \sqrt{ 1 - 476 \eps + \Oh(\eps^2) } \le 1. $$
\end{proof}

Now we use these curves in the OR-gadget from the last section. To this end, again partition the set of all assignments of $V_k$ into sets $A_k^1,\ldots,A_k^\ell$ of size $\Theta(2^{\nphi/2}/\ell)$, where we fix $1 \le \ell \le 2^{\nphi/2}$ later. Use the above construction of $P_1,P_2$ after replacing $A_1$ by $A_1^{j_1}$ and $A_2$ by $A_2^{j_2}$ for any $j_1,j_2 \in [\ell]$ to obtain curves $P_1^{j_1j_2}, P_2^{j_1j_2}$. Slightly rename these curves so that we have curves $(P_1^j,P_2^j)$ for $j \in [\ell^2]$. Then these curves satisfy \propref{PG}.

\begin{lemma} \label{lem:lasttt}
  The curves $P_1^j,P_2^j$ satisfy \propref{PG} with $c = \Theta(1+\sqrt \eps \mphi |A_k|)$ and $\beta = 1+\eps$. Moreover, $|P_k^j| = \Theta( \mphi 2^{\nphi/2} / \ell)$ for any $j \in [\ell^2]$, $k \in \{1,2\}$.
\end{lemma}
\begin{proof}
  Using \lemref{last}, we can follow the proof in \secref{discrete}, since everything that we used about $P_1,P_2$ is captured by this lemma. This proves that if $\varphi$ is satisfiable then $\ddfr(P_1^j,P_2^j) \le 1$ for some $j \in [\ell^2]$, and if $\varphi$ is not satisfiable then $\ddfr(P_1^j,P_2^j) > 1+\eps$ for all $j \in [\ell^2]$, i.e., \propref{PG}.(i) and (ii) in the discrete case. The same adaptations as in \secref{nondiscrete} allow to prove correctness in the continuous case, we omit the details.
  
  It is easy to see that all constructed points lie within distance 1 of $(0,0,0,0,0)$, showing~(iv). For (v) we use that we placed the points along the upper quarter-circle, and not the full circle. This way, all points in $P_1^j$ have a distance to $(0,\rho,0,0,0)$ of at most $\|(0,\rho) - (\tfrac12,\tfrac12)\| + \Oh(\sqrt \eps) < 1$, for sufficiently small $\eps$. 
  
  For (iii) observe that all segments of $P_k^j$ (except for the finitely many segments incident to $s_2^*,t_2^*$) have length $\Theta(\sqrt \eps + 1/(\mphi |A_k^j|))$, $k \in \{1,2\}$. Moreover, the $\Theta(\mphi |A_k^j|)$ segments of $P_k^j$ are spread along a quarter-circle. Hence, any ball $B(q,r)$ intersects $\Oh(1 + \min\{1, r\} \mphi |A_k^j|)$ segments of $P_k^j$. Since each of these segments has length $\Oh(\min\{r, \sqrt \eps + 1/(\mphi |A_k^j|)\})$ in $B(q,r)$, the total length of $P_k^j$ in $B(q,r)$ is $\Oh(r(1+\sqrt \eps \mphi |A_k^j|))$. Thus, $P_k^j$ is $\Oh(1+\sqrt \eps \mphi |A_k^j|)$-packed. It is also $\Theta(1+\sqrt \eps \mphi |A_k^j|)$-packed, since all $\Theta(\mphi |A_k^j|)$ segments of length $\Theta(\sqrt \eps + 1/(\mphi |A_k^j|))$ lie in a ball of radius 1 around $(0,0,0,0,0)$ or $(0,\rho,0,0,0)$ by (iv) and (v). Finally, note that $|A_k^j| = 2^{\nphi/2}/\ell$.
\end{proof}

\paragraph{Proof of \thmref{cpackedFive}}
The above \lemref{lasttt} allows to apply \lemref{OR}, which constructs curves $R_1,R_2$ such that any $(1+\eps)$-approximation for the \Fr distance of $(R_1,R_2)$ decides satisfiability of $\varphi$. Since $R_1$ and $R_2$ are $c$-packed with
$$ c = \Theta(1 + \sqrt \eps \mphi 2^{\nphi/2} / \ell), \quad n = \max\{|R_1|,|R_2|\} = \Theta(\ell \mphi 2^{\nphi/2}), $$
we obtain that any $(1+\eps)$-approximation for the \Fr distance with runtime $\Oh((c n / \sqrt \eps)^{1-\delta})$ yields an algorithm for \CNFSAT\ with runtime
$\Oh( \mphi^2 2^{(1-\delta)\nphi} )$,
as long as $\ell = \Oh(\sqrt \eps \mphi 2^{\nphi/2})$. This contradicts \CNFSETH. 

Moreover, using \lemref{seth} we can assume that $1 \le \mphi \le 2^{\delta \nphi / 4}$. Setting $\ell := \Theta(\eps^{\frac 1{2(1+\gamma)}} 2^{\frac{1-\gamma}{1+\gamma} \nphi/2})$ for any $0 \le \gamma \le 1$, we obtain
\begin{align*}
  \eps^{\frac 1{2(1+\gamma)}} 2^{\frac{2}{1+\gamma} \nphi /2} \le \,&n \le \eps^{\frac 1{2(1+\gamma)}} 2^{(\frac{2}{1+\gamma} + \delta/2) \nphi /2},  \\
  \eps^{\frac \gamma{2(1+\gamma)}} 2^{\frac{2\gamma}{1+\gamma} \nphi/2} \le \,&c \le \eps^{\frac \gamma{2(1+\gamma)}} 2^{(\frac{2\gamma}{1+\gamma} + \delta/2) \nphi/2}.
\end{align*}
From this it follows that $\Omega(n^{\gamma - \delta/2}) \le c \le \Oh(n^{\gamma+\delta})$, which implies the desired polynomial restriction $n^{\gamma - \delta} \le c \le n^{\gamma + \delta}$ for sufficiently large $n$.
Note that this works as long as 
$$1 \le \ell \le \Oh(\sqrt \eps \mphi 2^{\nphi/2}).$$ 
Since $\ell = \Theta\big( \big( \sqrt \eps n / c \big)^{1/2} \big)$, the first inequality is equivalent to 
$cn/\sqrt \eps \le n^2$, which is a natural condition, since otherwise the exact algorithm for general curves is faster.
Plugging in the definition of $\ell = \Theta(\eps^{\frac 1{2(1+\gamma)}} 2^{\frac{1-\gamma}{1+\gamma} \nphi/2})$, the second inequality becomes $1/\eps \le \big(2^{\nphi} M^{(1+\gamma)/\gamma}\big)^2$. Since $(1+\gamma)/\gamma \ge 2$, $n = \Oh(\ell \mphi 2^{\nphi/2}) \le \Oh( M^2 2^\nphi)$, and $c \ge 1$, this is implied by the first condition $c n / \sqrt \eps \le n^2$. Hence, we may choose any sufficiently small $\eps = \eps(n)$ with $cn/\sqrt \eps \le n^2$.

\section{Conclusion} \label{sec:conclusion}

We presented strong evidence that the (continuous or discrete) \Fr distance has no strongly subquadratic algorithms, by relating this problem to the Strong Exponential Time Hypothesis.

Our extensions of this main result include approximation algorithms and realistic input curves ($c$-packed curves). 
These extensions leave three particularly interesting open questions, asking for new algorithms or improved lower bounds. Here, we use $\tilde \Oh$ to ignore any polylogarithmic factors in $n$, $c$, and $1/\eps$.

\begin{enumerate}
  \item Is there a strongly subquadratic $\Oh(1)$-approximation for the \Fr distance on general curves? 
  \item In any dimension $d \in \{2,3,4\}$, is there a $(1+\eps)$-approximation with runtime $\tilde \Oh(c n)$ for the \Fr distance on $c$-packed curves? Or is there even an exact algorithm with runtime $\tilde \Oh(c n)$?
  \item In any dimension $d \ge 5$, is there a $(1+\eps)$-approximation with runtime $\tilde \Oh(c n / \sqrt{\eps})$ for the \Fr distance on $c$-packed curves?
\end{enumerate}

\paragraph{Acknowledgements}
The author wants to thank Wolfgang Mulzer for introducing him to the problem, Marvin K\"unnemann for useful discussions, Nabil H. Mustafa for his enthusiasm and encouragement, and an anonymous referee for helpful comments and pointers to the orthogonal vectors problem.

\end{document}